\documentclass[10pt, doublecolumn]{IEEEtran}
\usepackage{epsfig,latexsym}
\usepackage{float}
\usepackage{indentfirst}
\usepackage{amsmath}
\usepackage{bm}
\usepackage{amssymb}
\usepackage{times}
\usepackage{enumitem}

\usepackage{algorithm}
\usepackage[noend]{algpseudocode}
\usepackage{subfigure}
\usepackage{psfrag}
\usepackage{hyperref}
\usepackage{cite}
\usepackage{lastpage}
\usepackage{fancyhdr}
\usepackage{color} 
 \usepackage{amsthm}
\usepackage{bigints}
\usepackage{array}
\usepackage{booktabs}
\newtheorem{Lemma}{Lemma}
\newtheorem{Corollary}{Corollary}
\newtheorem{lemma}[Lemma]{$\mathbf{Lemma}$}

\newtheorem{corollary}[Corollary]{$\mathbf{Corollary}$}

\newcounter{problem}
\newcounter{save@equation}
\newcounter{save@problem}
\makeatletter
\newenvironment{problem}
{\setcounter{problem}{\value{save@problem}}%
  \setcounter{save@equation}{\value{equation}}%
  \let\c@equation\c@problem
  \subequations
}
{\endsubequations
  \setcounter{save@problem}{\value{equation}}%
  \setcounter{equation}{\value{save@equation}}%
}

\begin{document}
\title{  \vspace{-0.5em}\Huge{ Environment Division Multiple Access (EDMA): \\A Feasibility Study via Pinching Antennas  }}

\author{ Zhiguo Ding, \IEEEmembership{Fellow, IEEE}, Robert Schober, \IEEEmembership{Fellow, IEEE}, and H. Vincent Poor, \IEEEmembership{Life Fellow, IEEE}   \thanks{ 
  
\vspace{-1em}

Z. Ding is with the University
of Manchester, Manchester, M1 9BB, UK, and Khalifa University, Abu Dhabi, UAE.    
R. Schober is with the Institute for Digital Communications,
Friedrich-Alexander-University Erlangen-Nurnberg (FAU), Germany. H. V. Poor is  with the  Department of Electrical and Computer Engineering, Princeton University,
Princeton, NJ 08544, USA.
 

  }\vspace{-2.5em}}
 \maketitle

\begin{abstract}
This paper exploits the dynamic features of wireless propagation environments as the basis for a new multiple access technique, termed environment division multiple access (EDMA). In particular, with the proposed pinching-antenna-assisted EDMA, the multi-user propagation environment is intelligently reconfigured to improve signal strength at intended receivers and simultaneously suppress multiple-access interference, without requiring complex signal processing, e.g., precoding, beamforming, or multi-user detection. The key to creating a favorable propagation environment is to utilize the capability of pinching antennas to reconfigure line-of-sight (LoS) links, e.g., pinching antennas are placed at specific locations, such that interference links are blocked on purpose. Based on a straightforward choice of pinching-antenna locations, the ergodic sum-rate gain of EDMA over conventional multiple access and the probability that EDMA achieves a larger instantaneous sum rate than the considered benchmarking scheme are derived in closed form. The obtained analytical results demonstrate the significant potential of EDMA for supporting multi-user communications. Furthermore, pinching antenna location optimization is also investigated, since the locations of pinching antennas are critical for reconfiguring LoS links and large-scale path losses. Two low-complexity algorithms are developed for uplink and downlink transmission, respectively, and simulation results are provided to show their optimality in comparison to exhaustive searches.  
\end{abstract}\vspace{-0.2em}

\begin{IEEEkeywords}
Environment-division multiple access (EDMA), pinching antennas, antenna location optimization, line-of-sight blockage, uplink and downlink sum rates.  
\end{IEEEkeywords}
\vspace{-1em} 
\section{Introduction}
Multiple access techniques are a cornerstone of modern wireless communication systems, as they are key to ensuring that multiple users can be served efficiently given the scarce bandwidth resources  \cite{mojobabook}. There are two important metrics for measuring the success of a multiple access technique. One is the mitigation of multiple-access interference, and the other one is that each user is allocated as many bandwidth resources as possible. For example, by relying on orthogonal bandwidth resources available in the time and frequency domains,  time-division multiple access (TDMA), frequency-division multiple access (FDMA), and orthogonal frequency-division multiple access (OFDMA) can effectively suppress multiple-access interference, but each user can use only a portion of the available bandwidth resources \cite{Rappaport}.  By exploiting the code domain, code-division multiple access (CDMA) can effectively suppress multiple-access interference and also ensure that each user has access to the whole available bandwidth \cite{Verduebook,909605}. Similarly, by exploiting the spatial and power domains, the use of space-division multiple access (SDMA) and non-orthogonal multiple access (NOMA) can also ensure that multiple users can simultaneously use the same bandwidth resources  \cite{jsacnomaxmine}.


This paper is motivated by the question of whether there are additional degrees of freedom that can be exploited beyond the conventional time, frequency, code, spatial, and power domains, and proposes a new type of multiple access technique, termed environment division multiple access (EDMA). In particular, the key features of wireless propagation environments, such as line-of-sight (LoS) blockage and large-scale path losses, are shown in the paper to be useful for realizing multiple access. We note that conventionally, the wireless propagation environment was treated as fixed and uncontrollable, and only recently,  it has become a reconfigurable parameter, thanks to several emerging flexible-antenna techniques, including reconfigurable intelligent surfaces (RISs), movable antennas, fluid antennas, and pinching antennas \cite{irs1,irs2,10318061,9264694,pinching_antenna2,mypa}. Take pinching antennas as an example, which are well-known for their capability to intelligently reconfigure users' propagation environments and to create favorable large-scale path losses and small-scale multi-path fading gains \cite{pinching_antenna2,mypa,GPASS}. Recent studies have demonstrated that pinching antennas can be used to significantly improve system throughput, communication security, and sensing capability of wireless networks \cite{pamagazine,10976621,11169486,11202497,robertisac1,yarupass}.  Because of their capabilities to reconfigure wireless propagation environments, pinching antennas can also be exploited for new multiple access techniques. One example of EDMA is waveguide division multiple access (WDMA), where dielectric waveguide-based pinching antennas are used to realize a new form of SDMA, and so-called pinching beamforming is employed to suppress multiple-access interference \cite{wdma1}. Fluid antennas can also be leveraged for fluid-antenna multiple access (FAMA), where each user's fluid antenna is moved on a wavelength scale to create favorable small-scale multi-path fading for mitigating multiple-access interference \cite{9650760}. 

Unlike these existing multiple access techniques, this paper proposes a new type of EDMA by focusing on LoS blockage, which is a key feature of millimeter-wave (mmWave) and terahertz (THz) communications. In particular, in mmWave and THz networks, the LoS link between the transceivers is critical, and can be more than $20$ dB stronger than the non-LoS links \cite{6363891}. Although LoS blockage has traditionally been viewed as a harmful effect, it can in fact be exploited to effectively suppress co-channel interference \cite{11036558, kaidlos}. For example, in many communication scenarios with structured obstacles, e.g., tall shelving units in an Internet of Things (IoT) warehouse, as shown in Fig. \ref{fig1a}, the use of pinching antennas can ensure that interfering LoS links are effectively blocked. As a result, each user can enjoy interference-free transmission, and also has access to the whole available bandwidth. Furthermore, the use of LoS blockage for designing multiple access yields low system complexity, since the transceivers do not need to carry out complex signal processing, such as precoding, beamforming, multi-user detection, etc. To explore this idea further, the feasibility of using LoS blockage to realize EDMA in pinching-antenna systems is investigated in this paper, the contributions of which are as follows:

 \begin{figure}[t] \vspace{-0.2em}
\begin{center}
\subfigure[Potential application of EDMA]{\label{fig1a}\includegraphics[width=0.4\textwidth]{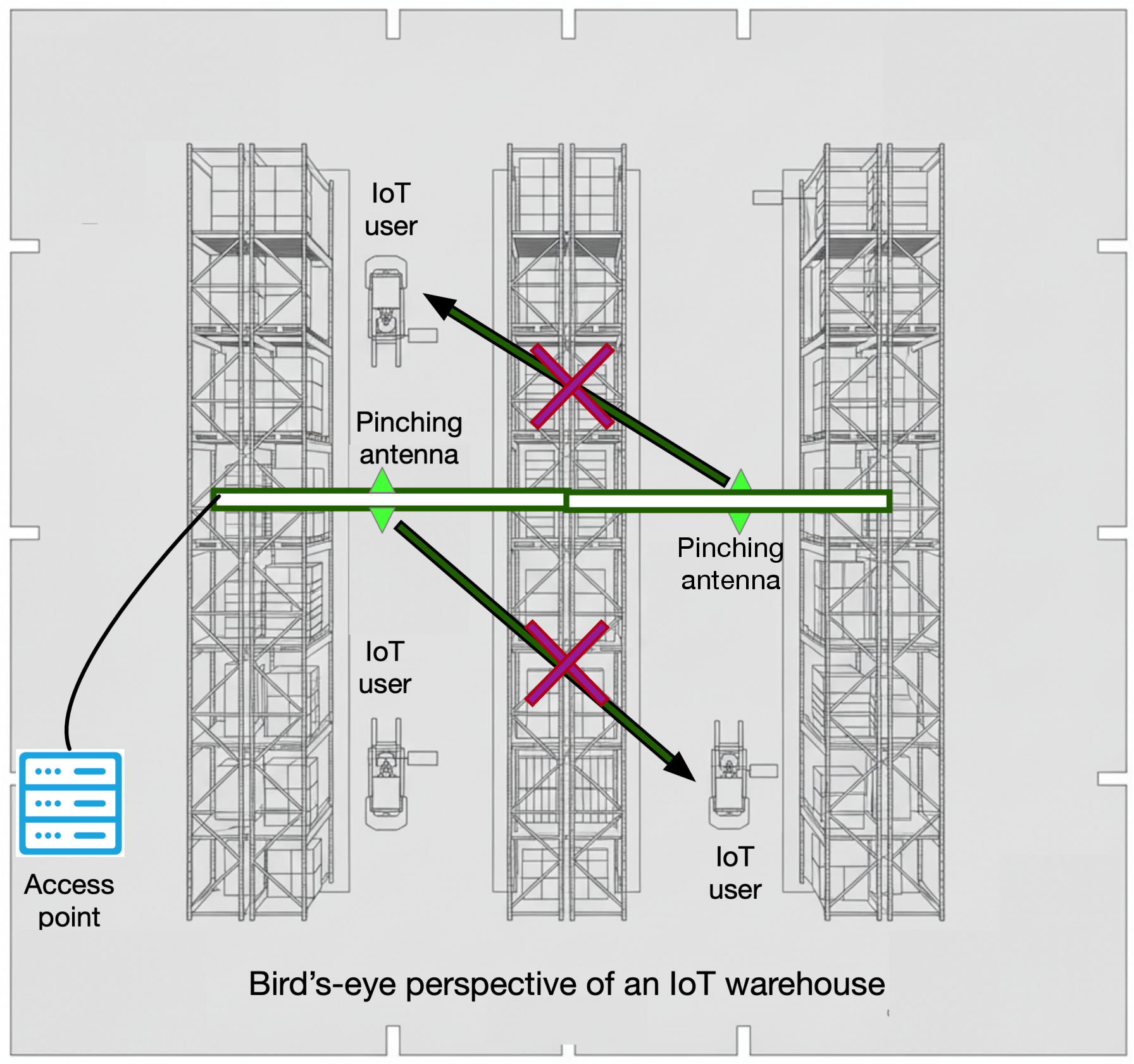}}  
\subfigure[System diagram of EDMA]{\label{fig1b}\includegraphics[width=0.40\textwidth]{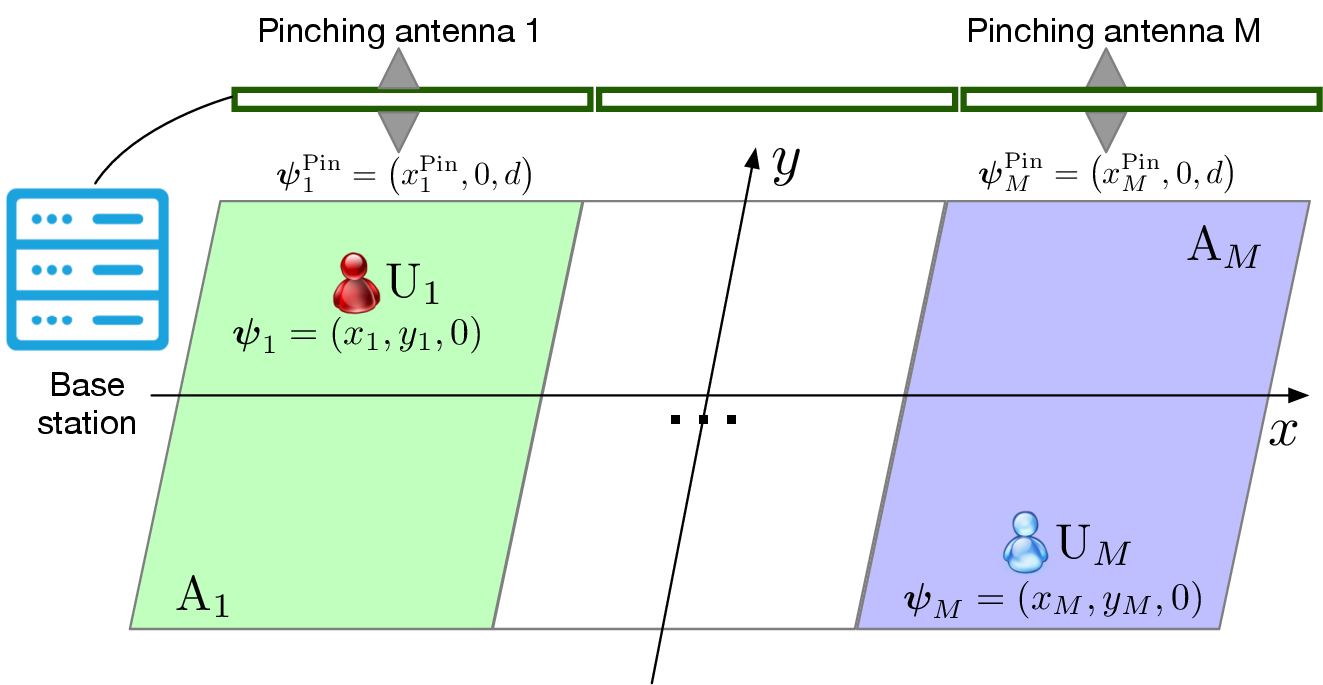}} 
\end{center}\vspace{-1em}
\caption{Illustration of the considered multi-user EDMA system. The Internet of Things (IoT) application figure in Fig. \ref{fig1a} is based on an image generated by Google Gemini.  \vspace{-1em} }\label{fig1}\vspace{-1.2em}
\end{figure}

\begin{itemize}
\item System models for both uplink and downlink EDMA are first developed in this paper. In particular, a base station is equipped with a segmented dielectric waveguide and serves multiple users simultaneously. On each waveguide segment, a single pinching antenna is activated to ensure that the following two purposes are simultaneously achieved. First, the pinching antenna is to be placed as close as possible to the user it serves, which is useful for reducing large-scale path losses and hence increases the signal strength for the considered uplink/downlink transmission. Second, the pinching antenna is to be placed in a position that ensures that its associated interference links suffer most likely from LoS blockage, in order to effectively mitigate multiple-access interference. 

\item The performance gain of EDMA over conventional multiple access techniques is investigated, where pinching-antenna-assisted TDMA is used as the benchmarking scheme, and the sum rate gain of EDMA over TDMA is used as the performance metric. To facilitate the performance analysis, the location of each pinching antenna is fixed based on the location of the user it serves. We note that this simple choice is not optimal, but has the benefit that the ergodic uplink and downlink sum rates of EDMA are identical. Due to the complex nature of EDMA performance analysis, a special case with two users is focused on, where a closed-form expression for the ergodic uplink/downlink sum-rate gain over TDMA can be obtained. In addition, an asymptotic study is carried out to establish the conditions under which EDMA achieves a larger ergodic sum rate than TDMA. Furthermore, a closed-form expression for the probability that EDMA achieves a larger instantaneous sum rate than TDMA is developed, which clearly demonstrates the impact of the service area size and the LoS blockage parameter on the performance gain of EDMA over TDMA. 

\item The locations of the pinching antennas are crucial to the performance of blockage-assisted EDMA, and hence pinching-antenna location optimization is also investigated in this paper. For the uplink case, an optimization problem for the joint optimization of the locations of the pinching antennas is formulated. Then, it is shown that the problem can be decomposed into multiple subproblems, where the locations of the pinching antennas can be individually optimized. Analytical results are developed to show that the objective function of each subproblem is a unimodal function. This observation motivates the proposed golden section search-based algorithm, which is shown to be capable of achieving the optimal performance. Unlike the uplink case, the locations of the pinching antennas have to be jointly designed for the downlink case. The convexity and concavity of the constraints of the formulated optimization problem are first analyzed, which facilitates the design of the proposed successive convex approximation (SCA) based algorithm. Simulation results are provided to show that the proposed SCA algorithm can achieve the same performance as an exhaustive search. 
\end{itemize}


\vspace{-1em}

\section{System Model}\label{section model}
In this paper, a multi-user pinching-antenna assisted EDMA system is considered, as shown in Fig. \ref{fig1b}. The key idea of this new multiple access technique is to leverage the capability of pinching antennas to reconfigure wireless propagation environments and partition the overall service area into isolated regions. In this way, a user served in one region does not interfere with others, which means that multiple users can be served simultaneously in an interference-free manner. In particular, assume that the service area is divided into $M$ smaller non-overlapping regions, denoted by ${\rm A}_m$, $1\leq m\leq M$, \cite{mypa}. The base station is equipped with a segmented dielectric waveguide, i.e., the waveguide is divided into $M$ segments. It is assumed that each segment is equipped with its own feed and a dedicated radio-frequency chain, which ensures that signals received by different segments do not interfere with each other \cite{Chongjunpa1xdd}. Without loss of generality, assume that a single user is scheduled in each of the small regions, and served by a single pinching antenna activated on the corresponding segment. Denote the user scheduled in ${\rm A}_m$ by ${\rm U}_m$, and its corresponding pinching antenna by  ${\rm PA} _m^{\rm Pin}$, $1\leq m \leq M$. Assume that ${\rm U}_m$ is uniformlly distributed within  ${\rm A}_m$.

Denote the sides of the rectangular-shaped service area by $D_{\rm L}$ and $D_{\rm W}$, respectively. A Cartesian coordinate system is used, where the center location of the service area is denoted by ${\boldsymbol \psi} _0=(0,0,0)$, the location of  ${\rm PA} _m^{\rm Pin}$ is denoted by ${\boldsymbol \psi} _m^{\rm Pin}=(x_m^{\rm Pin},0,d)$,  ${\rm U}_m$'s location is denoted by ${\boldsymbol \psi} _m=(x_m,y_m,0)$, $x_m\leq x_{m+1}$, and $d$ denotes the height of the waveguide. For illustration purposes, the sizes of all ${\rm A}_m$, $1\leq m \leq M$, are identical, which means that $-\frac{{\rm D}_W}{2}\leq y_m\leq \frac{{\rm D}_W}{2}$, and $x_m^s\leq y_m\leq x_m^e$, where $x_m^s$ and $x_m^e$ denote the boundaries of ${\rm A}_m$ along the x-axis in Fig. \ref{fig1b}.

The uplink signal received by  ${\rm PA} _m^{\rm Pin}$ is given by
 \begin{align}\label{models}
  y_m^{\rm Pin} =&   \alpha_{mm} {h}_{mm}   s_m +\sum_{i\neq m}  \alpha_{im}{h}_{im}  s_i +w_m^{\rm Pin} ,
  \end{align}
  where $s_m$ denotes the signal sent by  ${\rm U}_m$, $w_m^{\rm Pin}$ denotes the additive white Gassian noise at  ${\rm PA} _m^{\rm Pin}$, $h_{im}$ denotes the channel between ${\rm U}_i$ and ${\rm PA} _m^{\rm Pin}$ modelled as follows: 
 \begin{align}
 h_{im}=\frac{\sqrt{\eta} e^{-2\pi j \left(\frac{  1}{\lambda}\left| {\boldsymbol \psi}_i  - {\boldsymbol \psi}_m^{\rm Pin}\right|
  +\frac{1}{\lambda_g}\left| {\boldsymbol \psi}_{0m}^{\rm Pin}  - {\boldsymbol \psi}_m^{\rm Pin}\right|
  \right)}}{  \left| {\boldsymbol \psi} _i - {\boldsymbol \psi}_m^{\rm Pin}\right|} .
  \end{align}
Here,  ${\boldsymbol \psi}_{0m}^{\rm Pin} $ denotes the location of the feed point of the $m$-th waveguide segment, $\eta = \frac{c^2}{16\pi^2 f_c^2 }$, $c$ is the speed of light, $f_c$ is the carrier frequency, $\lambda=\frac{c}{f_c}$, $\lambda_g$ denotes the waveguide wavelength, and $\alpha_{im}$ is the associated indicator function for LoS blockage, e.g., $\alpha_{im}=1$ if there is an LoS link between ${\rm U}_i$ and ${\rm PA} _m^{\rm Pin}$, and $\alpha_{im}=0$ otherwise. In particular, this paper adopts the following LoS blockage model for ultra-dense indoor environments \cite{3gppblock}:
\begin{align}\label{blockage2}
 \mathbb{P}(\alpha_{im}=1) = e^{-\phi | \boldsymbol{\psi}_i-\boldsymbol{\psi}^{\rm Pin}_m |^2 } ,
\end{align}
where $\phi$ is the LoS blockage system parameter.   
  Therefore, ${\rm U}_m$'s achievable uplink data rate via EDMA is given by
  \begin{align}\label{eq4}
  R_m^{\rm UL} =& \log_2\left(
1+ \frac{ \alpha_{mm}\left|h_{mm} \right|^2}
{\sum_{i\neq m}  \alpha_{im}\left|
h_{im}   
\right|^2+\frac{1}{\rho}}
\right),
\end{align}
  where each user is assumed to use the same transmit power, and $\rho$ denotes the transmit signal-to-noise ratio (SNR). Similarly,  ${\rm U}_m$'s achievable downlink data rate can be obtained as follows: $  R_m^{\rm DL} = \log_2\left(
1+ \frac{ \alpha_{mm}\left|h_{mm} \right|^2}
{\sum_{i\neq m}  \alpha_{mi}\left|
h_{mi}   
\right|^2+\frac{1}{\rho}}
\right)$
  
 {\it Remark 1:} The use of segmented waveguides has two benefits. First, maintenance costs are reduced, e.g., if one segment is broken, only that single segment needs to be replaced, instead of the entire waveguide. Second, signals received by different segments are naturally separated, which increases the available degrees of freedom. However, the use of segmented waveguides requires each segment to be equipped with a dedicated feed, which increases system complexity. We note that an alternative is to use multiple waveguides without the constraint that they are on a straight line, which yields further flexibility for the system design but is beyond the scope of this paper.

      \begin{figure}[t]\centering \vspace{-0.2em}
    \epsfig{file=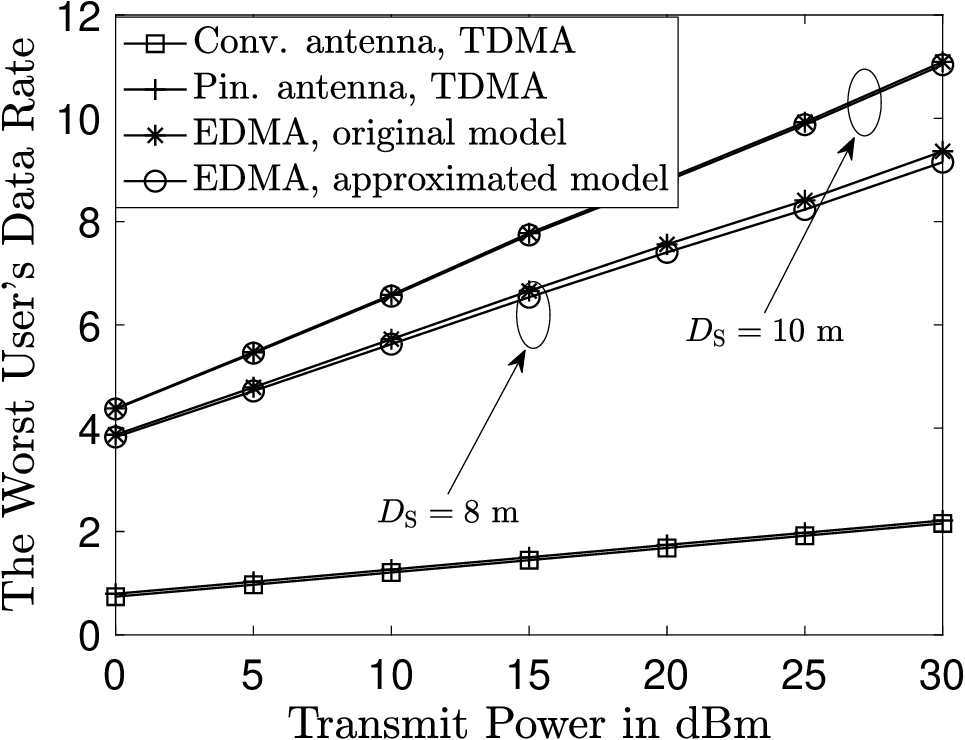, width=0.4\textwidth, clip=}\vspace{-0.5em}
\caption{Illustration of the accuracy of the approximated system model. The cases with and without taking into account interfering pinching antennas that are more than two segments away are shown. An uplink scenario with $M=5$ users is considered, $D_{\rm S}$ denotes the segment width, $D_{\rm W}=10$ m, $d=3$ m, $f_c=28$ GHz, $\phi=0.02$, and $D_{\rm L}=MD_{\rm S}$,   and the noise power is $-90$ dBm. The figure focuses on the worst user, ${\rm U}_3$, which suffers the most interference.
  \vspace{-1em}    }\label{fig2}   \vspace{-0.5em} 
\end{figure}

  \section{Sum-Rate Performance Analysis}\label{section analysis}
The aim of this section is to identify the performance gain of EDMA over pinching-antenna-assisted conventional TDMA, where the sum rate is used as the performance metric. 

To facilitate the performance analysis, we assume that the LoS link between a user and its dedicated pinching antenna always exists, i.e., $\alpha_{mm}=1$, which can be justified by the small distance between the user and the pinching antenna, as shown in Fig. \ref{fig1a}. We further assume that the LoS link between a user and an interfering pinching antenna which is more than two segments away does not exist due to the large distance\footnote{We note that for the cases with $M$, $M\leq 3$, there is no need for this assumption. We also note that the adopted assumption is analogous to the FDMA case, where only the interference from adjacent channels needs to be taken into account. }. Fig. \ref{fig2} shows the sum rates achieved with and without taking into account those interfering pinching antennas that are more than two segments away. As can be seen from Fig. \ref{fig2}, the used approximation results in negligible performance loss.

Therefore, by adopting the aforementioned approximation and the LoS blockage model shown in \eqref{blockage2}, for a user which has two neighboring users on each side, its uplink data rate achieved by EDMA can be simplified as follows: 
   \begin{align}\nonumber
   {R}_m^{\rm UL} &=   \mathbb{P}_{m-1}\mathbb{P}_{m+1}   \log_2\left(
1+ \frac{  \left|h_{mm} \right|^2}
{ \left|
h_{(m-1)m}   
\right|^2+\left|
h_{(m+1)m}   
\right|^2+\frac{1}{\rho}}
\right)\\\nonumber 
&+ \mathbb{P}_{m-1}\left(1- \mathbb{P}_{m+1} \right) \log_2\left(
1+ \frac{  \left|h_{mm} \right|^2}
{ \left|
h_{(m-1)m}   
\right|^2 +\frac{1}{\rho}}
\right)\\\nonumber
&+\left(1- \mathbb{P}_{m-1} \right)\mathbb{P}_{m+1}  \log_2\left(
1+ \frac{  \left|h_{mm} \right|^2}
{  \left|
h_{(m+1)m}   
\right|^2+\frac{1}{\rho}}
\right)
\\\label{eq5} &+
 \left(1- \mathbb{P}_{m-1} \right) \left(1-\mathbb{P}_{m+1}\right) \log_2\left(
1+ \rho \left|h_{mm} \right|^2
\right),
\end{align} 
where $\mathbb{P}_{m-1}=e^{-\phi | \boldsymbol{\psi}_{m-1}-\boldsymbol{\psi}^{\rm Pin}_{m}|^2 } $, and $\mathbb{P}_{m+1}=e^{-\phi | \boldsymbol{\psi}_{m+1}-\boldsymbol{\psi}^{\rm Pin}_{m}|^2 } $. The expression for the EDMA downlink data rate can be obtained similarly.   For users that have a single neighbouring user, e.g., ${\rm U}_1$ and ${\rm U}_M$, their data rates can be obtained in a straightforward manner.

To facilitate the performance analysis, a simple intuitive choice of the pinching antenna location is used, i.e., $\boldsymbol{\psi}^{\rm Pin}_{m} =(x_m,0,d)$. This choice can be justified by the fact that if the users are sufficiently separated, $ {R}_m^{\rm i}  \approx
  \log_2\left(
1+ \rho \left|h_{mm} \right|^2
\right)$, which means that $\boldsymbol{\psi}^{\rm Pin}_{m} =(x_m,0,d)$ is the optimal location choice. In the next section, more sophisticated choices of  $\boldsymbol{\psi}^{\rm Pin}_{m} $ will be considered.

With $\boldsymbol{\psi}^{\rm Pin}_{m} =(x_m,0,d)$, ${\rm U}_m$'s achievable uplink and downlink data rates can be expressed explicitly as shown in \eqref{eq5xx} and \eqref{eq5xxd} at the top of the next page. 
\begin{figure*}\vspace{-2em}
{\small  \begin{align}\nonumber
   {R}_m^{\rm UL} =&    e^{-\phi\left((x_m-x_{m-1})^2+y_{m-1}^2+d^2\right) }    e^{-\phi\left((x_m-x_{m+1})^2+y_{m+1}^2+d^2\right) }   \log_2\left(
1+ \frac{   \frac{\rho\eta}{y_m^2 +d^2}}
{   \frac{\rho\eta}{(x_m-x_{m-1})^2+ y_{m-1}^2 +d^2}+ \frac{\rho\eta}{(x_m-x_{m+1})^2+ y_{m+1}^2 +d^2}+1}
\right)\\\nonumber 
&+
 \left(1-e^{-\phi\left((x_m-x_{m-1})^2+y_{m-1}^2+d^2\right) }\right)    e^{-\phi\left((x_m-x_{m+1})^2+y_{m+1}^2+d^2\right) }   \log_2\left(
1+ \frac{   \frac{\rho\eta}{y_m^2 +d^2}}
{   \frac{\rho\eta}{(x_m-x_{m+1})^2+ y_{m+1}^2 +d^2}+1}
\right)
\\\nonumber
&+ e^{-\phi\left((x_m-x_{m-1})^2+y_{m-1}^2+d^2\right) }   \left(1- e^{-\phi\left((x_m-x_{m+1})^2+y_{m+1}^2+d^2\right) }  \right) \log_2\left(
1+ \frac{   \frac{\rho\eta}{y_m^2 +d^2}}
{   \frac{\rho\eta}{(x_m-x_{m-1})^2+ y_{m-1}^2 +d^2} +1}
\right)
\\\label{eq5xx} &+
  \left(1-e^{-\phi\left( (x_m-x_{m-1})^2+y_{m-1}^2+d^2 \right)}\right) \left(1- e^{-\phi\left((x_m-x_{m+1})^2+y_{m+1}^2+d^2\right) }  \right)  \log_2\left(
1+  \frac{\rho\eta}{y_m^2 +d^2}
\right),\\\nonumber
   {R}_m^{\rm DL} =&    e^{-\phi\left((x_m-x_{m-1})^2+y_{m}^2+d^2\right) }    e^{-\phi\left((x_m-x_{m+1})^2+y_{m}^2+d^2\right) }   \log_2\left(
1+ \frac{   \frac{\rho\eta}{y_m^2 +d^2}}
{   \frac{\rho\eta}{(x_m-x_{m-1})^2+ y_{m}^2 +d^2}+ \frac{\rho\eta}{(x_m-x_{m+1})^2+ y_{m}^2 +d^2}+1}
\right)\\\nonumber 
&+
 \left(1-e^{-\phi\left((x_m-x_{m-1})^2+y_{m}^2+d^2\right) }\right)    e^{-\phi\left((x_m-x_{m+1})^2+y_{m}^2+d^2\right) }   \log_2\left(
1+ \frac{   \frac{\rho\eta}{y_m^2 +d^2}}
{   \frac{\rho\eta}{(x_m-x_{m+1})^2+ y_{m}^2 +d^2}+1}
\right)
\\\nonumber
&+ e^{-\phi\left((x_m-x_{m-1})^2+y_{m}^2+d^2\right) }   \left(1- e^{-\phi\left((x_m-x_{m+1})^2+y_{m}^2+d^2\right) }  \right) \log_2\left(
1+ \frac{   \frac{\rho\eta}{y_m^2 +d^2}}
{   \frac{\rho\eta}{(x_m-x_{m-1})^2+ y_{m}^2 +d^2} +1}
\right)
\\\label{eq5xxd} &+
   \left(1-e^{-\phi\left( (x_m-x_{m})^2+y_{m-1}^2+d^2 \right)}\right)  \left(1- e^{-\phi\left((x_m-x_{m+1})^2+y_{m}^2+d^2\right) }  \right)\log_2\left(
1+  \frac{\rho\eta}{y_m^2 +d^2}
\right).
\end{align}}\vspace{-2em} 
\end{figure*} 
The aim of this section is to identify the sum-rate performance gain of EDMA over TDMA, which is defined as follows:
\begin{align}
\mathcal{E}\left\{\Delta^{\rm i}\right\}\triangleq &   \mathcal{E}\left\{R^{\rm EDMA}_{\rm sum,i}\right\} -  \mathcal{E}\left\{R^{\rm TDMA}_{\rm sum,i}  \right\},
 \end{align}
where ${\rm i} \in\{{\rm UL},{\rm DL}\}$, $\mathcal{E}\left\{\cdot\right\}$ denotes the expectation operation with respect to the random channel gains, $R^{\rm EDMA}_{\rm sum,i}=\sum^{M}_{m=1} {R}_m^{\rm i} $. It is straightforward to show that the uplink and downlink sum rates achieved by pinching-antenna-assisted TDMA are identical, and can be expressed as follows:
 \begin{align}
 R^{\rm TDMA}_{\rm sum,i} =\frac{1}{M}\sum^{M}_{m=1}
 \log_2\left(
1+ \rho \frac{\eta}{y_m^2 +d^2}
\right),
 \end{align}
  where ${\rm i} \in\{{\rm UL},{\rm DL}\}$, and the factor $\frac{1}{M}$ is needed since each user has access to a portion of the whole time duration only. 

By studying the difference between the two sum-rate expressions shown in \eqref{eq5xx} and \eqref{eq5xxd}, and also using the assumption that $y_m$, $1\leq m \leq M$, are independent and identically distributed (i.i.d.), the following corollary can be obtained in a straightforward manner.
\begin{corollary}\label{corollary1}
Assming that $y_m$, $1\leq m \leq M$, are i.i.d., the following equality holds:
 \begin{align}
\mathcal{E}\left\{R^{\rm EDMA}_{\rm sum,UL}\right\} =\mathcal{E}\left\{R^{\rm EDMA}_{\rm sum,DL}\right\} .
 \end{align}
\end{corollary}

Because of the uplink-downlink equivalence shown in Corollary \ref{corollary1}, the uplink case is focused on in the following, where the subscript, ${\rm UL}$, is omitted to simplify the notations. For the case with large $\phi$ and $D_{\rm L}$, it is likely that interference links suffer from LoS blockage, which means that the first three terms in \eqref{eq5xx} are negligible and hence the sum rate achieved by EDMA can be approximated as follows:
  \begin{align}\nonumber
    R^{\rm EDMA}_{\rm sum, LB} =  &\sum^{M}_{m=1}
 \left(1- e^{-\phi\left((x_m-x_{m+1})^2+y_{m+1}^2+d^2\right) }  \right) \\\nonumber &\times \left(1-e^{-\phi\left( (x_m-x_{m-1})^2+y_{m-1}^2+d^2 \right)}\right)\\  &\times  \log_2\left(
1+ \rho \frac{\eta}{y_m^2 +d^2}
\right). \label{approx}
 \end{align}
The approximation is a lower bound on $    R^{\rm EDMA}_{\rm sum}$, and becomes tight for the case with large $\phi$ and $D_{\rm L}$. The accuracy of this approximation will be studied in Section \ref{section simulations}.

Therefore, the instantaneous uplink sum-rate gain of EDMA over TDMA, $  R^{\rm EDMA}_{\rm sum} -  R^{\rm TDMA}_{\rm sum} $, can be lower bound as follows: 
   \begin{align}
\Delta_{\rm sum}^{\rm LB} = &   \sum^{M}_{m=1} \log_2\left(
1+ \rho \frac{\eta}{y_m^2 +d^2}
\right)\\  &\times \nonumber 
\left[\left(1- e^{-\phi \left((x_m-x_{m+1})^2+y_{m+1}^2+d^2\right) }  \right)\right.\\\nonumber &\left. \left(1-e^{-\phi \left((x_m-x_{m-1})^2+y_{m-1}^2+d^2\right) }\right)-\frac{1}{M}\right].
 \end{align}

It is challenging to develop a closed-form expression for the ergodic sum-rate gain, as it is a function of a large number of random variables, which motivates the study of the special cases in the following subsections.

{\it Remark 2:} The equivalence between the uplink and downlink cases shown in Corollary \ref{corollary1} is mainly due to the adopted choice of the pinching antenna locations. With different choices of the antenna locations, the uplink and downlink sum rates can be different, as shown in Section \ref{section optimization}.  
  \vspace{-1em}
\subsection{Special Case with $M=2$}
For the case with two users, the lower bound on the uplink sum-rate gain of EDMA over TDMA can be simplified as follows: 
 \begin{align}
\Delta_{\rm sum}^{\rm LB}
=&   \left(\frac{1}{2}-e^{-\phi\left(\left( x_2-x_1\right)^2+y_2^2 +d^2\right) } \right)
   \log_2\left(
1+ \rho \frac{\eta}{ y_1^2 +d^2}
\right)\\\nonumber &+ \left(\frac{1}{2}-e^{-\phi\left(\left(x_2-x_1\right)^2+y_1^2 +d^2\right) } \right)
   \log_2\left(
1+ \rho \frac{\eta}{ y_2^2 +d^2}
\right).
 \end{align}

 It is interesting to note that $\Delta_{\rm sum}^{\rm LB}$ depends on the difference between $x_1$ and $x_2$, rather than on the individual values of $x_1$ and $x_2$. 
Unlike \cite{mypa} where the users are i.i.d. deployed in the same service area, for the case considered in this paper, $x_1$ and $x_2$ are not identically distributed, e.g., $x_1^s\leq x_1\leq x_1^e\leq x_2^s\leq x_2\leq x_2^e$. However, by using the fact that the two users are uniformly distributed in $\rm{A}_1$ and $\rm{A}_2$, respectively, and also following steps similar to those in \cite{mypa}, the probability density function (pdf) of $x_2-x_1$ can be obtained as follows: 
\begin{align}\label{z1z2}
f_{x_2-x_1}(z) =
\begin{cases}
\frac{4z}{D_{\rm L}^2}, & 0 \le z \le \frac{D_{\rm L}}{2}\\[6pt]
\frac{4(D_{\rm L} - z)}{D_{\rm L}^2}, & \frac{D_{\rm L}}{2} \le z \le D_{\rm L}
\end{cases}.
\end{align}

Therefore, the lower bound on the ergodic uplink sum-rate gain of EDMA over TDMA can be expressed as follows:
    \begin{align}
\bar{\Delta}_{\rm sum}^{\rm LB}=  &2\left(\frac{1}{2}-e^{-\phi d^2  } \int^{\frac{D_{\rm W}}{2}}_{-\frac{D_{\rm W}}{2}}    e^{-\phi  y_2^2  } \frac{1}{D_{\rm W}}dy_2 T_1\right)\\\nonumber &\times
\int^{\frac{D_{\rm W}}{2}}_{-\frac{D_{\rm W}}{2}}    \log_2\left(
1+ \rho \frac{\eta}{ y_1^2 +d^2}
\right)\frac{1}{D_{\rm W}}dy_1 ,
 \end{align}
 where $T_1= \int^{D_{\rm L}}_{0}e^{-\phi z^2 }  f_{x_2-x_1}(z)dz$, and the fact that $y_1$ and $y_2$ are i.i.d. is used. 
 
 By using the pdf of $x_2-x_1$, $T_1$ can be evaluated as follows:
 \begin{align}\nonumber
T_1 
  =& \int^{\frac{D_{\rm L}}{2}}_{0}e^{-\phi z^2 } \frac{4z}{D_{\rm L}^2}dz  +
  \int_{\frac{D_{\rm L}}{2}}^{D_{\rm L}}e^{-\phi z^2 }  \frac{4(D_{\rm L} - z)}{D_{\rm L}^2} dz 
  \\\nonumber
  =&  \frac{2\sqrt{\pi}}{D_{\rm L}\sqrt{\phi}} \left[\Phi(\sqrt{\phi} D_{\rm L})-\Phi \Big(\frac{\sqrt{\phi} D_{\rm L}}{2}\Big)\right]\\ \label{T1} &
+\frac{2}{\phi D_{\rm L}^2} \left(1-2e^{-\frac{\phi D_{\rm L}^2}{4}}+e^{-\phi D_{\rm L}^2}\right),
   \end{align}
 where $\Phi(\cdot)$ denotes the probability integral \cite{GRADSHTEYN}.  Similarly, the expectation of $e^{-\phi  y_2^2  }$ with respect to $y_2$ can be evaluated as follows: 
 \begin{align}\label{T2}
 \int^{\frac{D_{\rm W}}{2}}_{-\frac{D_{\rm W}}{2}}    e^{-\phi  y_2^2  } \frac{1}{D_{\rm W}}dy_2 = \frac{\sqrt{\pi}}{D_{\rm W}\sqrt{\phi}} \Phi\left( \frac{\sqrt{\phi}D_{\rm W}}{2}\right)  .
 \end{align}
 
Define  $ g_1(a) = \int^{\frac{D_{\rm W}}{2}}_{0}\log_2\left(
 y ^2+a 
\right) dy$, whose closed-form expression can be obtained as follows:\cite{mypa}
  \begin{align}\label{ga}
  g_1(a) =   & \tau_2 -\log_2(e)D_{\rm W}+2\log_2(e)\sqrt{a} \tan^{-1}\left(\frac{D_{\rm W}}{2\sqrt{a}}\right),
  \end{align}
  where $\tau_2=\frac{D_{\rm W}}{2}\log_2\left(
 \frac{D_{\rm W}^2}{4}+a
\right) $, and $\tan(\cdot)^{-1}$ denotes the inverse tangent function.

By combining \eqref{T1}, \eqref{T2}, and \eqref{ga}, the following lemma can be obtained.
 \begin{lemma}\label{lemma1}
The lower bound on the uplink ergodic sum-rate gain of EDMA over TDMA is given by 
  \begin{align}
\bar{\Delta}_{\rm sum}^{\rm LB}
=&  \frac{2}{D_{\rm W}}
 \left(1-2e^{-\phi d^2  } T_1 \frac{\sqrt{\pi}}{D_{\rm W}\sqrt{\phi}} \Phi\left( \frac{\sqrt{\phi}D_{\rm W}}{2}\right)   \right)
   \\\nonumber &\times \left[g_1\left( d^2+\rho\eta\right)-g_1\left( d^2\right)\right].
 \end{align}
\end{lemma}
 Lemma \ref{lemma1} facilitates the asymptotic result provided in the following lemma. 
  \begin{lemma}\label{lemma2}
For the two-user uplink case with ${\phi} D_{\rm W}^2\rightarrow \infty$ and $ {\phi} D_{\rm L}^2\rightarrow \infty$,  $\bar{\Delta}_{\rm sum}^{\rm LB} \geq 0$.  For the case with $ {\phi} D_{\rm W}^2\rightarrow 0$ and $ {\phi} D_{\rm L}^2\rightarrow 0$,  $\bar{\Delta}_{\rm sum}^{\rm LB}\leq 0$.
\end{lemma}
 \begin{proof}
 See Appendix \ref{prooflemma2}.
 \end{proof}

{\it Remark 3:}   Lemma \ref{lemma2} illustrates the ideal application scenario for  EDMA. In particular, if the service region is large and the LoS blockage is severe, EDMA is guaranteed to outperform TDMA.   In Section \ref{section simulations}, simulation results are presented to show that, even for small or moderate values of $\phi$ and the service area size, EDMA can still outperform TDMA, particularly if pinching antenna location optimization is carried out.  

{\it Remark 4:} We note that a positive ergodic sum-rate gain does not guarantee that the instantaneous sum-rate gain is always positive. The latter can be characterized by the probability that the instantaneous sum rate of EDMA is larger than that of TDMA, and is to be studied in the next subsection.    
  \vspace{-1em}
 \subsection{Special Case with $M=2$ and $y_m=0$}
For the special case with $M=2$ and $y_m=0$, concise analytical results can be obtained, as shown in the following. We note that this special case is important in practice, since users in tunnel and railway communication scenarios are expected to be underneath the waveguide, i.e., $y_m=0$.  In particular, by applying Lemma \ref{lemma1}, a simplified expression for $\bar{\Delta}_{\rm sum}^{\rm LB}$ is obtained in the following corollary.
 \begin{corollary}
 For the considered special case, the lower bound on the ergodic downlink/uplink sum-rate performance gain of EDMA over TDMA can be expressed as follows:
   \begin{align} \label{crollay1eq}
\bar{\Delta}_{\rm sum}^{\rm LB}
= &    \left(1-4  e^{-\phi d^2 }   \frac{\sqrt{\pi}}{D_{\rm L}\sqrt{\phi}} \left[\Phi(\sqrt{\phi} D_{\rm L})-\Phi \Big(\frac{\sqrt{\phi} D_{\rm L}}{2}\Big)\right] \right. \\\nonumber &\left. 
+\frac{4  e^{-\phi d^2 } }{\phi D_{\rm L}^2} \left(1-2e^{-\frac{\phi D_{\rm L}^2}{4}}+e^{-\phi D_{\rm L}^2}\right) \right) 
   \log_2\left(
1+ \rho \frac{\eta}{ d^2}
\right)  .
 \end{align}
 
 \end{corollary}
 
Following the steps in the proof for Lemma \ref{lemma2}, it is straightforward to show that  for the case with $\phi D_{\rm L}^2\rightarrow \infty$, $\bar{\Delta}_{\rm sum}^{\rm LB}$ shown in \eqref{crollay1eq} can be approximated as follows:
   \begin{align}    \label{dfe1}
\bar{\Delta}_{\rm sum}^{\rm LB}
\approx &   \left(1+\frac{4  e^{-\phi d^2 }}{\phi D_{\rm L}^2} \right)
   \log_2\left(
1+ \rho \frac{\eta}{ d^2}
\right) \geq 0 ,
 \end{align}
 which means that EDMA outperforms TDMA.  
On the other hand, assuming that $\phi D_{\rm L}^2\rightarrow 0$, $\bar{\Delta}_{\rm sum}^{\rm LB}$ shown in \eqref{crollay1eq} can be approximated as follows:
    \begin{align} \label{dfe2}
\bar{\Delta}_{\rm sum}^{\rm LB}
\approx &       -6  e^{-\phi d^2 }    
   \log_2\left(
1+ \rho \frac{\eta}{ d^2}
\right)\leq 0,
 \end{align}
 which means that TDMA outperforms EDMA.  Both \eqref{dfe1} and \eqref{dfe2} are consistent with the conclusions drawn in Lemma \ref{lemma2}.

In addition to the ergodic sum-rate gain, it is also of interest to study the probability that EDMA yields a larger instantaneous sum rate than TDMA, which is analyzed in the following. For the considered special case, the instantaneous sum-rate difference between the two considered schemes can be lower bounded as follows: 
  \begin{align}
\Delta_{\rm sum}  \geq  &     \left(1-2e^{-\phi\left(\left( x_2-x_1\right)^2 +d^2\right) } \right)
   \log_2\left(
1+ \rho \frac{\eta}{  d^2}
\right) .
 \end{align}

Therefore, the probability that EDMA outperforms TDMA can be expressed as follows: 
  \begin{align}
\mathbb{P} (\Delta_{\rm sum} \geq 0)\geq & \mathbb{P} \left(1-2e^{-\phi\left(\left( x_2-x_1\right)^2 +d^2\right) }\geq 0\right)\\\nonumber = & \mathbb{P} \left(    x_2-x_1   \geq \underset{\nu}{\underbrace{\sqrt{\frac{\log(2)}{\phi}-d^2}}}\right),
 \end{align}
 since $x_2\geq x_1$, 
 where it is assumed that $\frac{\log(2)}{\phi}-d^2\geq 0$. Otherwise, EDMA always outperforms TDMA. Furthermore, assume that $\nu\leq D_{\rm L}$, since the user separation cannot be larger than the width of the overall service area. 
By applying the pdf of $ x_2-x_1$ shown in \eqref{z1z2} and with some straightforward algebraic manipulations, the following lemma is obtained.

\begin{lemma}\label{lemma3}
 For the considered special case, the probability for EDMA to realize a larger instantaneous sum rate than TDMA can be lower bounded as follows:
 \begin{align}\label{lemma3eq}
 \mathbb{P} (\Delta_{\rm sum} \geq 0)\geq \left\{\begin{array}{cc} 
1- \frac{2\nu^2}{D_{\rm L}^2}, &    \nu\leq \frac{D_{\rm L}}{2}\\ 2  \left(1 -\frac{\nu}{D_{\rm L}}\right)^2 , &   \frac{D_{\rm L}}{2}< \nu\leq D_{\rm L}\end{array}\right..
 \end{align}
 
\end{lemma}

{\it Remark 5:}  Lemma \ref{lemma3} shows that the lower bound on $ \mathbb{P} (\Delta_{\rm sum}\geq 0)$ is always a monotonically increasing function of $D_{\rm L}$, for both the cases shown in \eqref{lemma3eq}. In other words, the performance gain of EDMA over TDMA can be improved by increasing $D_{\rm L}$. This conclusion is expected since increasing $D_{\rm L}$ increases user separation and hence improves the performance gain of EDMA over TDMA.    
 
 \section{Pinching Antenna Location Optimization }\label{section optimization}
In the previous section, a simple configuration of pinching antenna locations was used to facilitate performance analysis. We note that the key feature of pinching antennas is their flexibility to be activated at any desired position along the waveguide. Therefore, it is important to study how the flexibility of pinching antennas can be used to further improve the performance of EDMA over TDMA.

In particular, the following pinching antenna location optimization problem is considered:
     \begin{problem}\label{pb:1} 
  \begin{alignat}{2}
\underset{x_m^{\rm Pin} }{\rm{max}}  &\quad   \min\left\{ {R}_1^{\rm i}, \cdots,   {R}_M^{\rm i}\right\} 
\\ s.t. &\quad  \label{1tst:1}  x_m^s\leq x_m^{\rm Pin}  \leq x_{m}^e , \quad 1\leq m \leq M,
  \end{alignat}
\end{problem}  
where ${\rm i}\in\{ {\rm UL}, {\rm DL} \}$, $\min\left\{ {R}_1^{\rm i}, \cdots,   {R}_M^{\rm i}\right\} $ is used as the objective function in order to guaranttee user fairness, and the constraint in \eqref{1tst:1} ensures that ${\rm U}_m$'s pinching antenna is activated within the $m$-th sgement. Pinching-antenna position optimization for uplink and downlink is studied in the following subsections separately, as their corresponding data rate expressions are different. 
%

\subsection{Uplink Pinching Antenna Location Optimization}
An important observation from the uplink data rate expression in \eqref{eq5xx} is that the choice of one pinching antenna has no impact on the other users' data rates. This observation can be used to decompose problem \ref{pb:1} into $M$ decoupled subproblems as follows:
     \begin{problem}\label{pb:2} 
  \begin{alignat}{2}
\underset{x_m^{\rm Pin} }{\rm{max}}  &\quad    {R}_m^{\rm UL} \label{1tst:2}
\\ s.t. &\quad    x_m^s\leq x_m^{\rm Pin}  \leq x_{m}^e.
  \end{alignat}
\end{problem}  
In other words, for the uplink case, the locations of the $M$ pinching antennas can be optimized in a parallel and decoupled manner.

\subsubsection{Properties of $ {R}_m^{\rm UL}$} Problem \ref{pb:2} is still challenging to solve, primarily due to the complex expression of $ {R}_m^{\rm UL}$. Similar to \eqref{approx},  problem \ref{pb:2} can be approximated as follows:
       \begin{problem}\label{pb:3} 
  \begin{alignat}{2}\nonumber
\underset{x_m^{\rm Pin} }{\rm{max}}  &\quad  f\left( x_m^{\rm Pin}\right)  \triangleq \left(1- e^{-\phi\left((x_{m+1}-x_m^{\rm Pin} )^2+y_{m+1}^2+d^2\right) }  \right) \\\nonumber &\quad\quad \times \left(1-e^{-\phi\left( (x_m^{\rm Pin}-x_{m-1})^2+y_{m-1}^2+d^2 \right)}\right)\\  &\quad\quad\times  \log_2\left(
1+ \rho \frac{\eta}{ (x_m^{\rm Pin}-x_{m})^2+ y_m^2 +d^2}
\right)
\\ s.t. &\quad   x_m^s\leq x_m^{\rm Pin}  \leq x_{m}^e .
  \end{alignat}
\end{problem}  

Although $f\left( x_m^{\rm Pin}\right)$ is much simpler than ${R}_m^{\rm UL}$, its expression is still very complex, which motivates the study of the special case with $y_m=0$, which reveals the unimodal property of $f\left( x_m^{\rm Pin}\right) $ in the range of $x_{m-1}\leq x_m^{\rm Pin}\leq x_{m+1}$.

\begin{lemma}\label{lemma4}
At high SNR, for the considered special case with $y_m=0$, a sufficient condition for $f\left( x_m^{\rm Pin}\right) $ to be a unimodal function of $x_m^{\rm Pin}$, for $x_{m-1}\leq x_m^{\rm Pin}\leq x_{m+1}$, is $\frac{\delta^2\phi}{2}  \geq 1$, where $\delta=\frac{x_{m+1}-x_{m-1}}{2}$.
\end{lemma}
\begin{proof}
See Appendix \ref{prooflemma4}.
\end{proof}

{\it Remark 6:} It can be straightforwardly verified that for moderately large $\phi$ and $D_{\rm L}$, the condition in Lemma \ref{lemma4} is satisfied. We note that for the case with small $\phi$ and $D_{\rm L}$, simulations have shown that $f\left( x_m^{\rm Pin}\right) $ might not be a unimodal function; however, the original function $ {R}_m^{\rm UL}$ is still most likely unimodal\footnote{We note that for the case that  ${R}_m^{\rm UL}$ is not a unimodal function, the use of the golden-section search algorithm yields a suboptimal solution only. However, as shown in Section \ref{section simulations}, for all the considered system parameters, the use of the golden-section search algorithm can achieve the same performance as an exhaustive search, i.e., the proposed algorithm realizes near-optimal performance with significantly reduced complexity.}. 
The global maximum of the considered unimodal function can be efficiently found with the golden section search algorithm presented in the next subsection. 

{\it Remark 7:} Assuming $\rho \rightarrow \infty$ and $y_m=0$, problem \ref{pb:3} can be approximated as follows:
       \begin{problem}\label{pb:3dx} 
  \begin{alignat}{2} 
\underset{x_m^{\rm Pin} }{\rm{max}}  &\quad  \left(1- e^{-\phi\left((x_{m+1}-x_m^{\rm Pin} )^2 +d^2\right) }  \right) \\\nonumber &\quad\quad \times \left(1-e^{-\phi\left( (x_m^{\rm Pin}-x_{m-1})^2 +d^2 \right)}\right)
\\ s.t. &\quad   x_m^s\leq x_m^{\rm Pin}  \leq x_{m}^e .
  \end{alignat}
\end{problem}  
Given the symmetric structure of the objective function of problem \ref{pb:3dx}, the optimal location for the pinching antenna is $x_m^{\rm Pin} =\frac{x_{m-1}+x_{m+1}}{2}$, if $x_m^s\leq \frac{x_{m-1}+x_{m+1}}{2} \leq x_{m}^e $. In Section \ref{section simulations}, simulation results will be provided to show that the solution of $x_m^{\rm Pin} =\frac{x_{m-1}+x_{m+1}}{2}$ is a surprisingly good choice for uplink EDMA, and can even outperform the solution of $x_m^{\rm Pin} =x_{m}$. 
\subsubsection{Golden section search algorithm}
The above unimodal discussion motivates the use of the golden section search algorithm outlined in Algorithm \ref{algorithm1}. In particular, as shown in Algorithm \ref{algorithm1}, the boundaries of the $m$-th segment are used as the initial boundaries of the search interval. During the $n$-th iteration, the boundaries of the current search interval, denoted by $x_{m}^{\rm UB}$ and $x_{m}^{\rm LB}$, are used to generate two new boundary point candidates, i.e., $x_{m,n}^{\rm LB}=x_m^{\rm UB}-\varphi \left(x_m^{\rm UB}-x_m^{\rm LB}\right)$ and 
$x_{m,n}^{\rm UB}=x_m^{\rm LB}+\varphi \left(x_m^{\rm UB}-x_m^{\rm LB}\right)$, where $n$ denotes the iteration index, and the golden ratio conjugate is given by $\varphi=\frac{\sqrt{5}-1}{2}$ \cite{7946258}. Based on the relationship between $R_m^{\rm UL}\left( x_{m,n}^{\rm UB}\right)$ and $R_m^{\rm UL}\left( x_{m,n}^{\rm LB}\right)$, the boundaries of the new search interval can be updated accordingly, as shown in Algorithm \ref{algorithm1}.

 \begin{algorithm}[t]
\caption{Golden Section Search Algorithm}

 \begin{algorithmic}[1]
 
\State Initial search interval: $x_m^{\rm LB}=x_m^s$, $x_m^{\rm UB}=x_m^e$
\State Set the tolerance threshold, $\epsilon$, define the golden ratio conjugate $\varphi$, and $n=1$
\If { $x_m^{\rm UB}-x_m^{\rm LB}>\epsilon$ } 
\State $n=n+1$
\State Generate two new boundary candidates:  

$x_{m,n}^{\rm LB}=x_m^{\rm UB}-\varphi \left(x_m^{\rm UB}-x_m^{\rm LB}\right)$, 

$x_{m,n}^{\rm UB}=x_m^{\rm LB}+\varphi \left(x_m^{\rm UB}-x_m^{\rm LB}\right)$
   \If { $R_m^{\rm UL}\left( x_{m,n}^{\rm UB}\right)>R_m^{\rm UL}\left( x_{m,n}^{\rm LB}\right)$ } 
   
   \State $x_m^{\rm UB}= x_{m,n}^{\rm UB}$

      \Else 
      \State $x_m^{\rm LB}= x_{m,n}^{\rm LB}$
     
  \EndIf
 \State \textbf{end}
 \EndIf
 
\State \textbf{end}
 \State The final antenna location is geneated as follows: $x_m^{\rm Pin}=\frac{x_m^{\rm UB}+x_m^{\rm LB}}{2}$.

 \end{algorithmic}\label{algorithm1}
\end{algorithm}

  \vspace{-1em}
\subsection{Downlink Pinching Antenna Location Optimization}
For the downlink case, problem \ref{pb:1} can be first equvalently recast as follows:
\begin{problem}\label{pb:4} 
  \begin{alignat}{2}
\underset{x_m^{\rm Pin} ,t }{\rm{max}}  &\quad   t 
\\ s.t. &\quad     R^{\rm DL}_m\geq t , \quad 1\leq m \leq M\\ &\quad x_m^s\leq x_m^{\rm Pin}  \leq x_{m}^e , \quad 1\leq m \leq M.
  \end{alignat}
\end{problem}  

Unlike the uplink case, the location of one downlink user's pinching antenna has a significant impact on the other users' data rates. As a result, the locations of all $M$ downlink users' pinching antennas need to be jointly optimized. 

Due to the complex expression of $R^{\rm DL}_m$, the approximation in \eqref{approx} is used, which means that problem \ref{pb:4} is approximated as follows: 
       \begin{problem}\label{pb:5} 
  \begin{alignat}{2}
\underset{x_m^{\rm Pin} , t}{\rm{max}}  &\quad   t \label{1tst:5}
\\ \nonumber s.t. &\quad    \left(1- e^{-\phi\left((x_{m-1}^{\rm Pin}-x_{m})^2+y_{m}^2+d^2\right) }  \right) \\&\times\nonumber  \left(1- e^{-\phi\left((x_{m+1}^{\rm Pin}-x_{m})^2+y_{m}^2+d^2\right) }  \right)\\  &\times  \log_2\left(
1+ \rho \frac{\eta}{(x_m^{\rm Pin}-x_m)^2+y_m^2 +d^2}
\right)\geq t,  \\ &\quad  x_m^s\leq x_m^{\rm Pin}  \leq x_{m}^e , \quad 1\leq m \leq M.
  \end{alignat}
\end{problem}  
 
By defining $u=\log(t\log 2)$, problem \ref{pb:5} can be further experessed as follows:
   \begin{problem}\label{pb:6} 
  \begin{alignat}{2}
\underset{x_m^{\rm Pin} , u}{\rm{max}}  &\quad   u \label{1tst:6}
\\  \label{2tst:6} s.t. &\quad  u - \log   \left(1- e^{-\phi\left((x_{m-1}^{\rm Pin}-x_{m})^2+y_{m}^2+d^2\right) }  \right) \\&-\nonumber \log \left(1- e^{-\phi\left((x_{m+1}^{\rm Pin}-x_{m})^2+y_{m}^2+d^2\right) }  \right)\\  &- \log \log\left(
1+ \rho \frac{\eta}{(x_m^{\rm Pin}-x_m)^2+y_m^2 +d^2}
\right)\leq0 , \nonumber \\ &\quad  x_m^s\leq x_m^{\rm Pin}  \leq x_{m}^e , \quad 1\leq m \leq M.
  \end{alignat}
\end{problem}

Problem \ref{pb:6} is non-convex, mainly due to its non-convex constraints \eqref{2tst:6}. In the following, the principle of SCA is applied to find a suboptimal solution of the considered optimization problem \cite{11036619,8502790}. 

To facilitate the implemenation of SCA, by defining $ f_{m1}(z)=\log \left(1- e^{-\phi\left(z^2+y_{m}^2+d^2\right) }  \right)$
and $f_{m2}(z)=\log \log\left(
1+ \rho \frac{\eta}{(z-x_m)^2+y_m^2 +d^2}
\right)$, problem \ref{pb:6} can be recast as follows:
   \begin{problem}\label{pb:7} 
  \begin{alignat}{2}
\underset{x_m^{\rm Pin} , u}{\rm{max}}  &\quad   u \label{1tst:7}
\\ \label{12tst:6}  s.t. &\quad  u -f_{m1}\left(x_{m-1}^{\rm Pin}-x_{m}\right)-f_{m1}\left(x_{m+1}^{\rm Pin}-x_{m}\right)  \\\nonumber &\quad -f_{m2}\left(x_{m}^{\rm Pin}\right)\leq0 ,  \\ &\quad  x_m^s\leq x_m^{\rm Pin}  \leq x_{m}^e , \quad 1\leq m \leq M.
  \end{alignat}
\end{problem}

For an efficient implementation of SCA, the convexity/concavity of the two functions, $f_{m1}(z)$ and $f_{m2}(z)$, is analyzed first. The first-order derivative of $f_{m1}(z)$ is given by 
 \begin{align} 
f_{m1}'(z) = \frac{2\phi z e^{-\phi (z^2 + y_m^2 + d^2)}}{1 - e^{-\phi (z^2 + y_m^2 + d^2)}},
\end{align}
and its second-order derivative is given by
\begin{align} 
f_{m1}''(z)
= 2\phi e^{-\phi(z^2 + y_m^2 + d^2)} 
\frac{1 - e^{-\phi(z^2 + y_m^2 + d^2)}-2\phi z^2}{(1 - e^{-\phi(z^2 + y_m^2 + d^2)})^2}
 .
\end{align}
The sign of $f_{m1}''(z)$ is mainly determined by  $f_{m3}(z^2) \triangleq  1 - e^{-\phi(z^2 + y_m^2 + d^2)}-2\phi z^2$. We note that $f_{m3}(x)=0$, $x\geq 0$, has a single root, as explained in the following. $f_{m3}(x)$ is a monotonically decreasing function of $x$, $x\geq0$, since its first-order derivative is strictly negative:
\begin{align}
f'_{m3}(x) =  \phi e^{-\phi(x + y_m^2 + d^2)}-2\phi <0. 
\end{align}
Given the fact that $f_{m3}(0)>0$ and $f_{m3}(\infty)\rightarrow -\infty$, $1 - e^{-\phi(z^2 + y_m^2 + d^2)}-2\phi z^2=0$ has a single root, which is denoted by $\beta_m$ and given by
\begin{align}
\beta^*=\sqrt{\frac{1+2W_{0}\!\left(-\tfrac12\,e^{-\frac{1}{2}-\phi(y_m^2+d^2)}\right)}{2\phi}},
\end{align}
where $W_0(\cdot)$ denotes the Lambert function \cite{6559999}. 
For our considered scenarios with moderate user separation, we observe that $\left|x_{m-1}^{\rm Pin}-x_{m}\right|\leq \beta^*$ and $\left|x_{m+1}^{\rm Pin}-x_{m}\right|\leq \beta^*$, which means that $f_{m1}(z)$ is a convex function. This motivates the following approximations:
\begin{align}
f_{m1}\left(x_{m-1}^{\rm Pin}-x_{m}\right)  &\approx \tilde{f}_{m1}\left(x_{m-1}^{\rm Pin},x_{m-1}^{\rm Pin,0} \right)  
,\\\nonumber
f_{m1}\left(x_{m+1}^{\rm Pin}-x_{m}\right)  &\approx \bar{f}_{m1}\left(x_{m+1}^{\rm Pin},x_{m+1}^{\rm Pin,0} \right)  ,
\end{align}
where $x_{m-1}^{\rm Pin,0} $ and $x_{m+1}^{\rm Pin,0} $ are the estimates from the previous iteration of the proposed SCA algorithm, 
$\tilde{f}_{m1}\left(x_{m-1}^{\rm Pin}-x_{m}\right)  \triangleq f_{m1}\left(x_{m-1}^{\rm Pin,0}-x_{m}\right)  +f_{m1}'\left(x_{m-1}^{\rm Pin}-x_{m}\right) \left(x_{m-1}^{\rm Pin}-x_{m-1}^{\rm Pin,0}\right)$ and $\bar{f}_{m1}\left(x_{m+1}^{\rm Pin}-x_{m}\right) \triangleq f_{m1}\left(x_{m+1}^{\rm Pin,0}-x_{m}\right) +f_{m1}'\left(x_{m+1}^{\rm Pin}-x_{m}\right) \left(x_{m+1}^{\rm Pin}-x_{m+1}^{\rm Pin,0}\right).$

On the other hand, we note that, at high SNR, $f_{m2}(z)$ is a concave function for most considered scenarios, and hence there is no need to apply SCA for the term $f_{m2}(z)$. In particular, at high SNR, $f_{m2}(z)$ can be approximated as follows:
\begin{align}
f_{m2}(z)\approx \log\left[\log \rho \eta -  \log\left( (z-x_m)^2+y_m^2 +d^2
\right)\right].
\end{align}
The first-order derivative of $f_{m2}(z)$ is given by
\begin{align}
f_{m2}'(z)= -\frac{2(z-x_m)}{\left((z-x_m)^2+y_m^2+d^2 \right)\log\big(\tfrac{\rho \eta}{(z-x_m)^2+y_m^2+d^2 }\big)},
\end{align}
and its second-order derivative is shown in \eqref{fm2z} at the top of this page, where the last step of \eqref{fm2z} follows by the high SNR approximation.

\begin{figure*}[!]\vspace{-2em}
{\small \begin{align} \label{fm2z}
f_{m2}''(z)
=&\frac{
4(z-x_m)^2\!\left(\log\!\dfrac{\rho\eta}{(z-x_m)^2+y_m^2+d^2}-1\right)
-2\left((z-x_m)^2+y_m^2+d^2\right)\!
\log\!\dfrac{\rho\eta}{(z-x_m)^2+y_m^2+d^2}
}{
\left((z-x_m)^2+y_m^2+d^2\right)^2\,
\left[\log\!\dfrac{\rho\eta}{(z-x_m)^2+y_m^2+d^2}\right]^2
} \\\nonumber
\approx &\frac{
\left(2(z-x_m)^2 
-y_m^2-d^2 \right)
\log\dfrac{\rho\eta}{(z-x_m)^2+y_m^2+d^2}
}{
\left((z-x_m)^2+y_m^2+d^2\right)^2\,
\left[\log\!\dfrac{\rho\eta}{(z-x_m)^2+y_m^2+d^2}\right]^2
}.
\end{align}}\vspace{-2em}
\end{figure*}

For most considered communication scenarios, it is expected that a pinching antenna is placed not too far away from the user it serves, i.e., 
 \begin{align}
 2\left(x_{m}^{\rm Pin}-x_m\right)^2 
-y_m^2-d^2\leq 0,
 \end{align}
 which means $f_{m2}(x_{m}^{\rm Pin})$ is a concave function of $x_{m}^{\rm Pin}$ and hence there is no need to apply SCA for this term.

In summary, the SCA-based algorithm is carried out in an iterative manner. Naturally, the initialization of the optimization parameters, denoted by $x_m^{\rm Pin, (0)}$, can be obtained as follows: $
x_m^{\rm Pin (0)}=x_m$. At the $k$-th iteration of the proposed SCA algorithm, by using the estimates of the antenna locations from the previous iteration, $x_m^{{\rm Pin}, (k-1)}$,  problem \ref{pb:5} can be approximated as follows:
   \begin{problem}\label{pb:9} 
  \begin{alignat}{2}
\underset{x_m^{\rm Pin} , u}{\rm{max}}  &\quad   u \label{1tst:9}
\\  \label{2tst:9} s.t. &\quad  u -\tilde{f}_{m1} \left(x_{m-1}^{\rm Pin},x_{m-1}^{{\rm Pin},(k-1)} \right)    \\\nonumber &\quad  -\bar{f}_{m1}\left(x_{m+1}^{\rm Pin},x_{m+1}^{{\rm Pin},(k-1)} \right) -f_{m2}\left(x_{m}^{\rm Pin}\right)\leq0 ,  \\ &\quad  x_m^s\leq x_m^{\rm Pin}  \leq x_{m}^e , \quad 1\leq m \leq M.
  \end{alignat}
\end{problem}  
For users which are at the boundary of the overall service area, they have a single neighbouring user, and the corresponding approximations can be straightforwardly obtained.

{\it Remark 8:} Recall that a key step of SCA is to construct convex surrogate functions for the original constraints \cite{11036619,8502790}.  If $f_{m1}(z)$ is convex and $f_{m2}(x)$ is concave, the approximated function in \eqref{2tst:9} is indeed an upper bound on the original constraint function in \eqref{12tst:6}, which ensures that the solution obtained in each iteration is still a feasible solution of problem \ref{pb:7}. When the service area is large or the SNR is low, $f_{m1}(z)$ may become non-convex, and $f_{m2}(x)$ may lose concavity. Nevertheless, our extensive computer simulations show that the proposed SCA algorithm remains robust across various system parameter settings and consistently achieves near-optimal performance.



 
  \begin{figure}[!] \vspace{-1em}
\begin{center}
\subfigure[Uplink, $\phi=0.02$ ]{\label{fig4a}\includegraphics[width=0.38\textwidth]{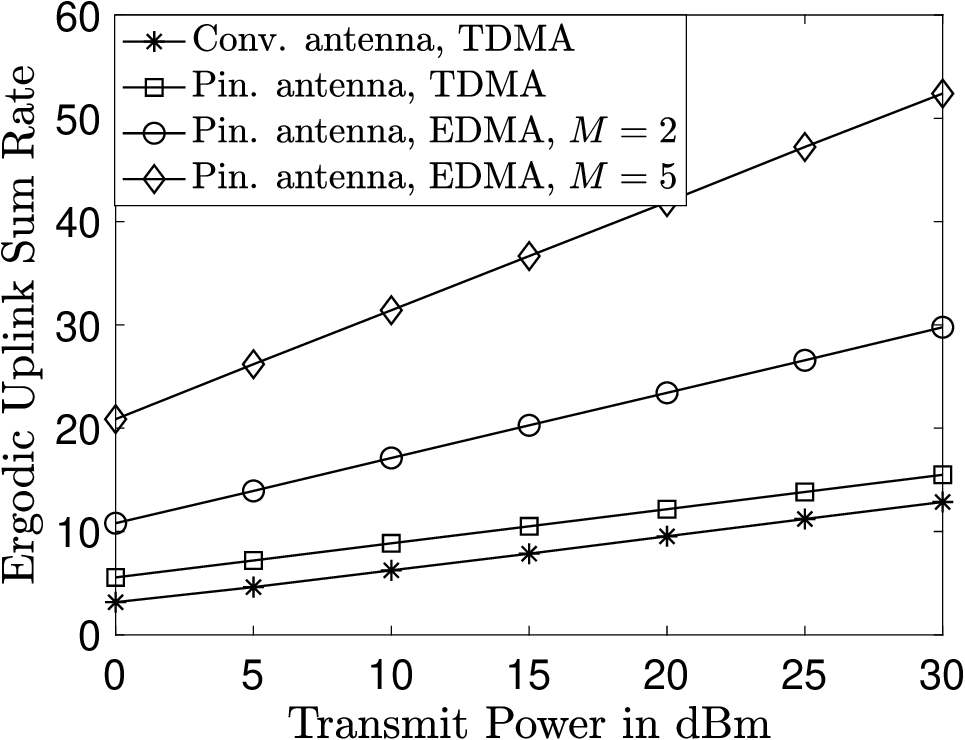}} 
\subfigure[Uplink, $\phi=0.04$ ]{\label{fig4b}\includegraphics[width=0.38\textwidth]{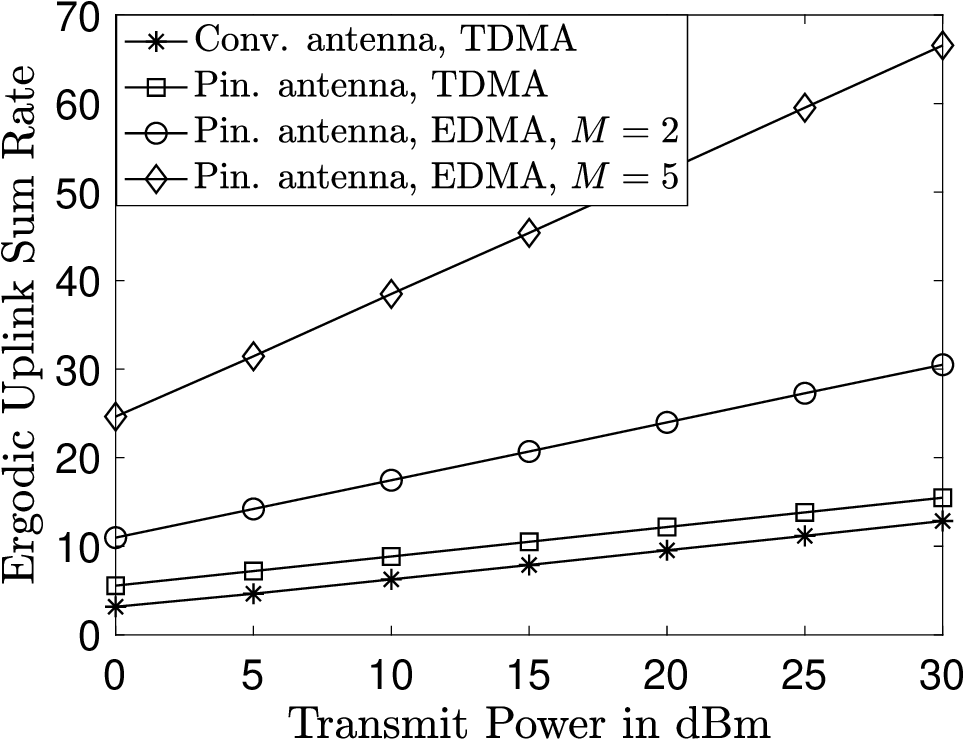}}  
\subfigure[Downlink, $\phi=0.02$ ]{\label{fig4c}\includegraphics[width=0.4\textwidth]{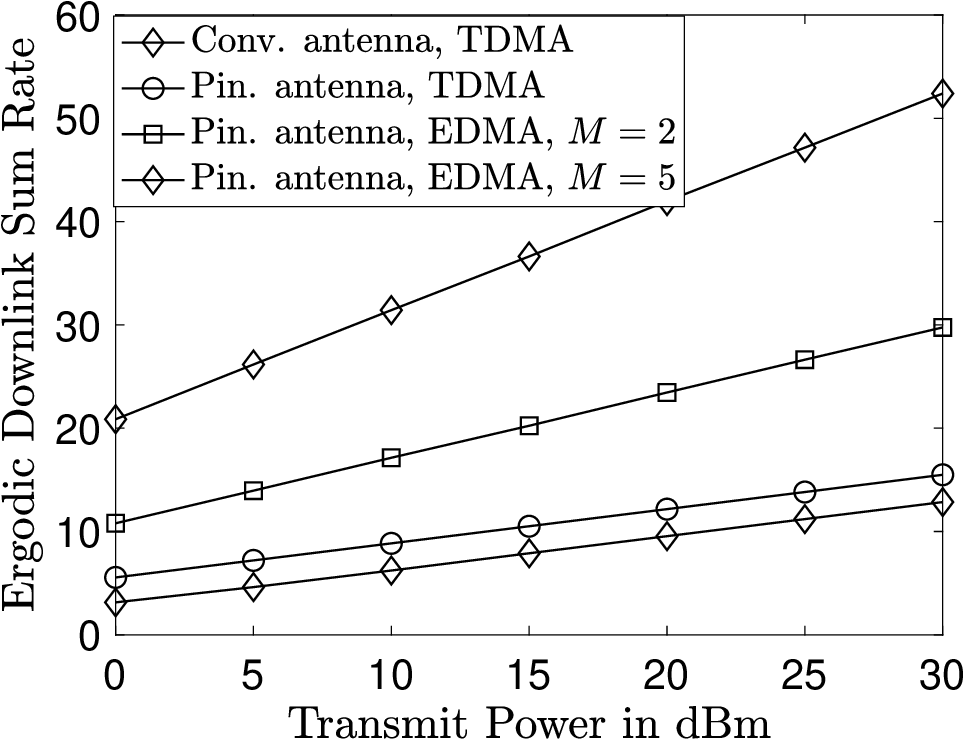}}\vspace{-1em}
\end{center}
\caption{ The uplink and downlink ergodic sum rates achieved by the proposed EDMA scheme, where the TDMA schemes assisted by conventional and pinching antennas are used as the benchmarking schemes, and $D_{\rm L} =4D_{\rm W} $.   \vspace{-1em} }\label{fig4}\vspace{-1.2em}
\end{figure}

\section{Numerical Results} \label{section simulations}
In this section, computer simulation results are presented to evaluate the performance of the proposed EDMA scheme as well as the accuracy of the developed analytical results. For all conducted simulations, $f_c=28$ GHz, $D_{\rm W}=10$ m, $d=3$ m, and the noise power is $-90$ dBm.

Fig. \ref{fig4} shows the uplink and downlink ergodic sum rates achieved by the proposed EDMA scheme, where $\boldsymbol{\psi}^{\rm Pin}_{m} =(x_m,0,d)$ is used for the pinching antenna locations. TDMA assisted by conventional and pinching antennas, respectively, is used as a benchmark to facilitate performance comparison. In particular, for conventional-antenna assisted TDMA, the antenna of the base station is deployed at the center of the service area, i.e., its location is $(0,0,d)$, and each user is served in a dedicated time slot. For pinching-antenna assisted TDMA, in the $m$-th time slot, a single pinching antenna is activated at $(x_m, 0, d)$ to serve ${\rm U}_m$. As can be seen from Fig. \ref{fig4a}, the use of EDMA yields a significant uplink sum rate gain over the TDMA benchmarking schemes, which is due to the fact that the $M$ users can be simultaneously served by EDMA. Comparing Fig. \ref{fig4a} to \ref{fig4b}, one can observe that the performance gain of EDMA over TDMA can be improved by increasing $\phi$. This observation is due to the fact that a large $\phi$ effectively suppresses co-channel interference, i.e., the LoS link between ${\rm U}_i$ and ${\rm PA}_m$, $m\neq i$, is likely to be blocked. We note that the ergodic downlink sum rate shown in Fig. \ref{fig4c} is identical to the ergodic uplink sum rate shown in Fig. \ref{fig4a}, which confirms Corollary \ref{corollary1}. Fig. \ref{fig4} also shows that the sum rate of EDMA increases by increasing the number of users from $2$ to $5$, which can be explained as follows. With EDMA, LoS blockages can be utilized to partition the service area into multiple small regions, within each of which a user can be served with suppressed multiple-access interference. However, it is important to point out that further increasing the number of users can decrease the sum rate of EDMA, as shown in Table \ref{table1}, where C-TDMA denotes conventional-antenna assisted TDMA, P-TDMA denotes pinching-antenna assisted TDMA, $\phi=0.02$, and the transmit power is $30$ dBm. This observation is due to the fact that multiple-access interference becomes severe when there are too many users in the service area.

 \begin{table}[!]
\centering
\caption{Impact of $M$ on the ergodic uplink sum rate of EDMA \vspace{-1em}}
\begin{tabular}{*7c}
\toprule 
&    $M=3$&$M=5$  & $M=7$ & $M=9$&$M=11$     \\
    \hline
   C-TDMA    &$4.87  $&$  3.01  $&$  2.18$ &$1.70  $&$    1.39  $ 
    \\
    \hline
    P-TDMA     &$5.15  $&$  3.09  $&$  2.21  $  &$1.72  $&$   1.40     $  \\
    \hline
  EDMA   &$41.07  $&$  51.80 $&$   50.28$ &$43.91  $&$  37.15    $  \\
\bottomrule
\end{tabular}\label{table1}\vspace{-1em}
\end{table}

     \begin{figure}[t]\centering \vspace{-0.2em}
    \epsfig{file=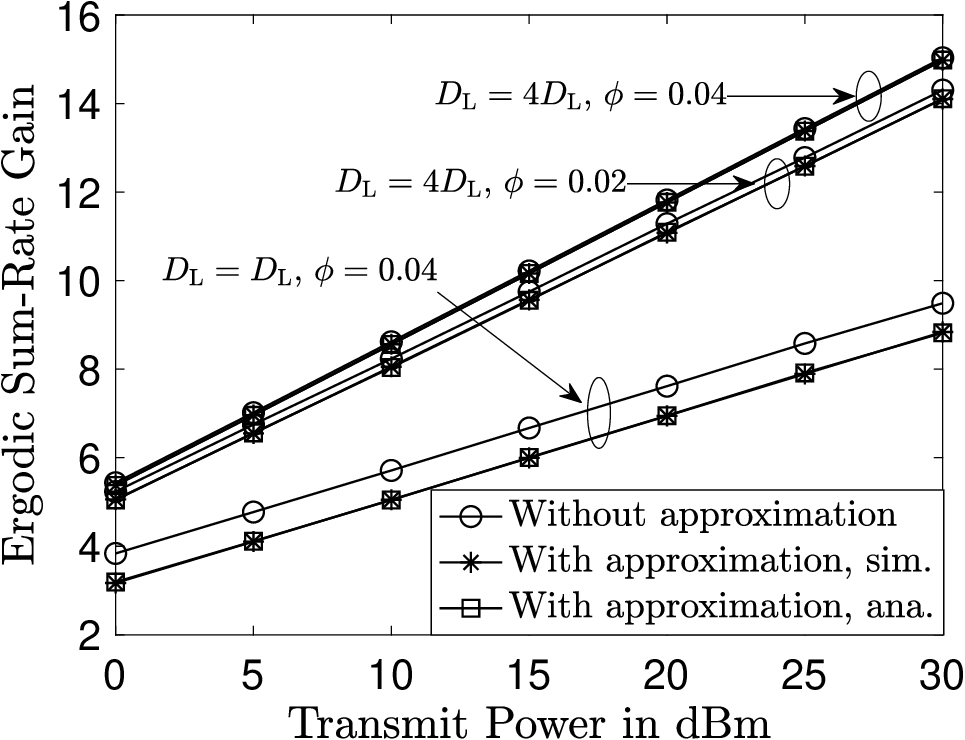, width=0.38\textwidth, clip=}\vspace{-0.5em}
\caption{The ergodic sum-rate gain achieved by EDMA over pinching-antenna assisted TDMA, where the special case with $M=2$ is focused on. The approximation is based on \eqref{approx}, and the shown analytical results are based on Lemma \ref{lemma1}. 
  \vspace{-1em}    }\label{fig5}   \vspace{-0.5em} 
\end{figure}

     \begin{figure}[!]\centering \vspace{-0.2em}
    \epsfig{file=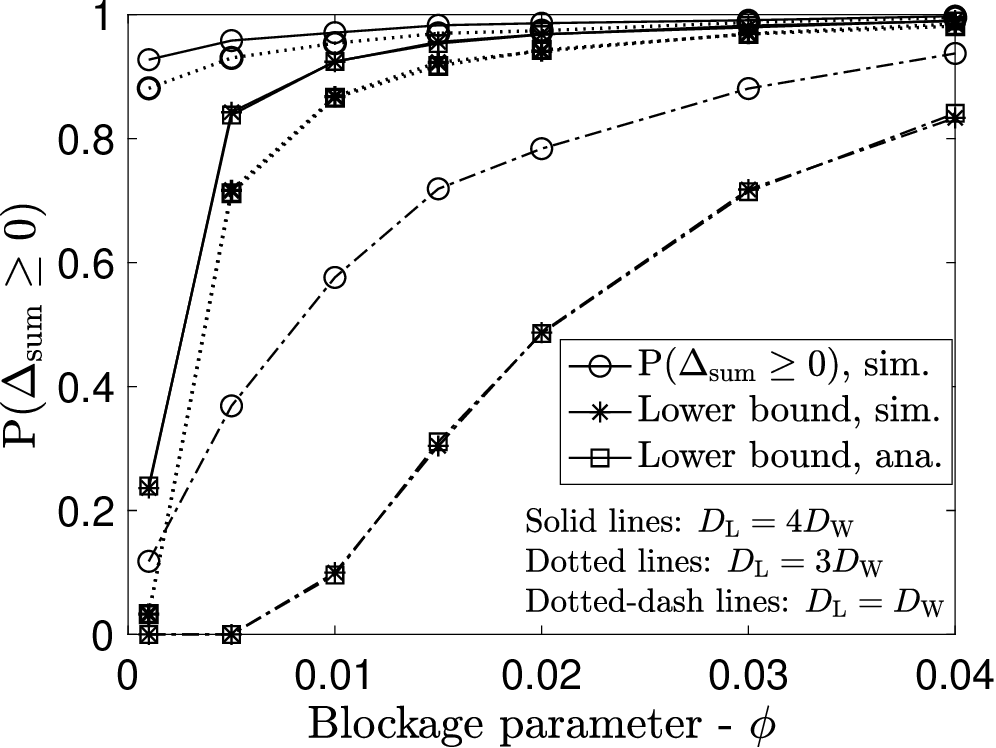, width=0.38\textwidth, clip=}\vspace{-0.5em}
\caption{The probability of EDMA outperforming pinching-antenna assisted TDMA, $\mathbb{P} (\Delta_{\rm sum}\geq 0)$, where $M=2$, $y_m=0$, $1\leq m\leq M$, and the transmit power is $30$ dBm. The analytical results are based on Lemma \ref{lemma3}.  
  \vspace{-1em}    }\label{fig6}   \vspace{-1em} 
\end{figure}

In Figs. \ref{fig5} and \ref{fig6}, the sum-rate gain of EDMA over TDMA, $\Delta_{\rm sum}$, is studied by focusing on the special case of $M=2$, where the antenna location for EDMA is give by $\boldsymbol{\psi}^{\rm Pin}_{m} =(x_m,0,d)$. In particular, Fig. \ref{fig5} shows the ergodic sum-rate gain as a function of the transmit power. As can be seen from the figure, the sum-rate gain of EDMA over TDMA is a monotonically increasing function of the transmit power, and can be further increased by increasing the size of the service area or the LoS blockage parameter $\phi$. In addition, Fig. \ref{fig5} shows that the approximation in \eqref{approx} is accurate for large $D_{\rm L}$ and $\phi$, and the fact that the analytical curve perfectly matches the simulated one verifies the accuracy of the analytical results in Lemma \ref{lemma1}. In Fig. \ref{fig6}, the probability for EDMA to achieve a larger sum rate than TDMA is investigated. As can be seen from the figure, increasing $D_{\rm L}$ or $\phi$ enhances the likelihood that EDMA outperforms TDMA, which is consistent with the observation made in Fig. \ref{fig5}. In addition, Fig. \ref{fig6} shows that the analytical and simulated curves coincide perfectly, which verifies the accuracy of the analytical results in Lemma \ref{lemma3}. 

   \begin{figure}[!] \vspace{-2em}
\begin{center}
\subfigure[ A general case with $M=5$ uniformly deployed users ]{\label{fig7a}\includegraphics[width=0.38\textwidth]{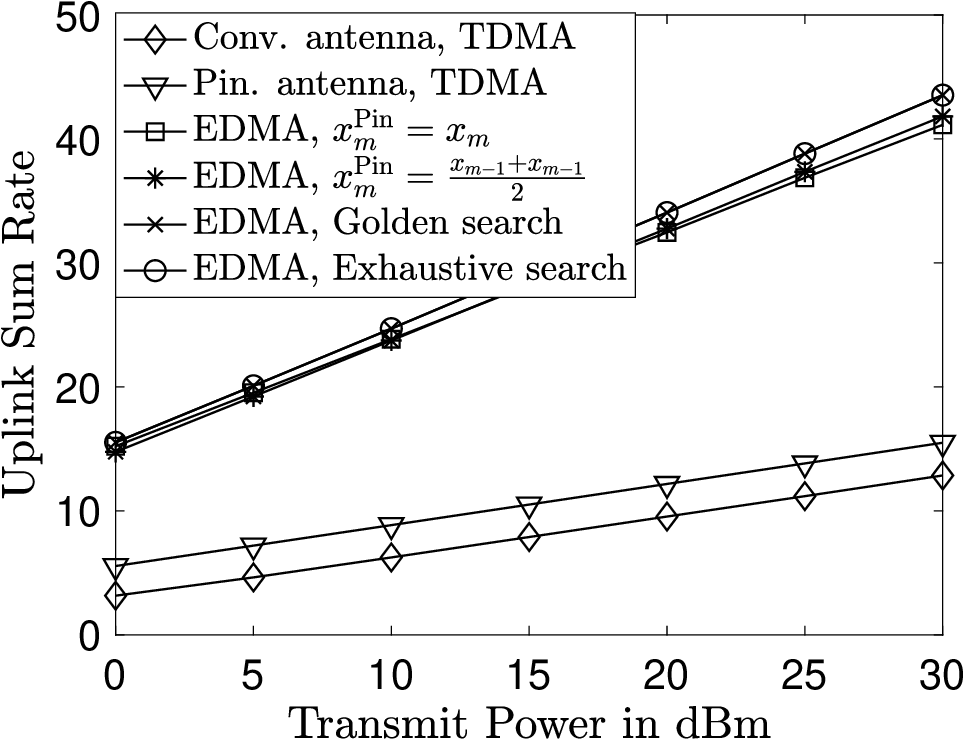}} 
\subfigure[A special case with $M=3$ clustered users ]{\label{fig7b}\includegraphics[width=0.38\textwidth]{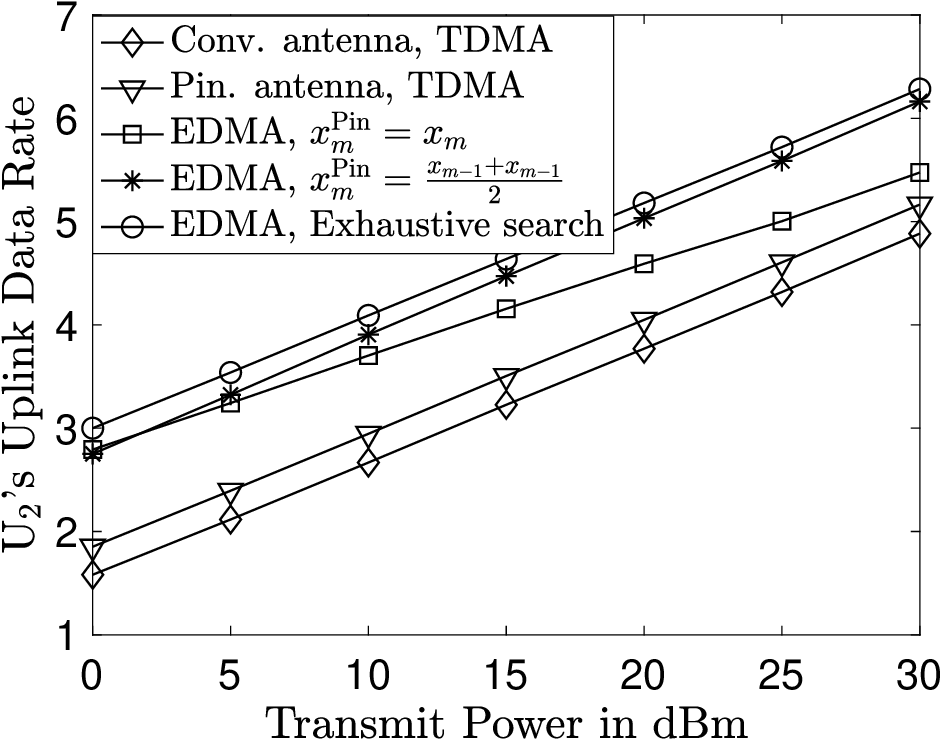}}   \vspace{-1em}
\end{center}
\caption{Impact of antenna location optimization on the uplink sum rate achieved by EDMA, where $D_{\rm L}=4D_{\rm W}$ and $\phi=0.02$. In Fig. \ref{fig7a}, the widths of ${\rm A}_i$, $i\in \{1,2,3\}$,  are identical, i.e., $\frac{D_{\rm L}}{M}$. In Fig. \ref{fig7b}, the widths of ${\rm A}_i$, $i\in \{1,3\}$,  are $\frac{D_{\rm L}}{4M}$, and the width of ${\rm A}_2$ is $\frac{D_{\rm L}}{M}$.  \vspace{-1em} }\label{fig7}\vspace{-1.2em}
\end{figure}

In Fig. \ref{fig7}, the impact of antenna location optimization on the uplink sum rate achieved by EDMA is investigated. In Fig. \ref{fig7a}, the sizes of all ${\rm A}_{m}$, $1\leq m \leq M$, are assumed to be identical, which means that the users are equally spaced in a statistical sense. In Fig. \ref{fig7b}, a clustered scenario is considered, where ${\rm U}_1$ and ${\rm U}_3$ are deployed in two narrow strips next to ${\rm U}_2$, and hence the three users can be very close. For both scenarios considered in Figs. \ref{fig7a} and \ref{fig7b},  the proposed golden section search algorithm realizes the same performance as an exhaustive search, which confirms the optimality of the proposed algorithm. If the users are equally spaced, the difference between the performance achieved by the two fixed choices, i.e., $x_m^{\rm Pin}=x_m$ and  $x_m^{\rm Pin}=\frac{x_{m-1}+x_{m+1}}{2}$, is small, as shown in Fig. \ref{fig7a}. However, if the users are close to each other, the use of   $x_m^{\rm Pin}=\frac{x_{m-1}+x_{m+1}}{2}$ can still yield a performance close to the optimal, but the choice of $x_m^{\rm Pin}=x_m$ can result in a signifiant performance loss at high SNR. 

   \begin{figure}[!] \vspace{-2em}
\begin{center}
\subfigure[$M=2$ and $D_{\rm L}=D_{\rm W}$ ]{\label{fig8a}\includegraphics[width=0.38\textwidth]{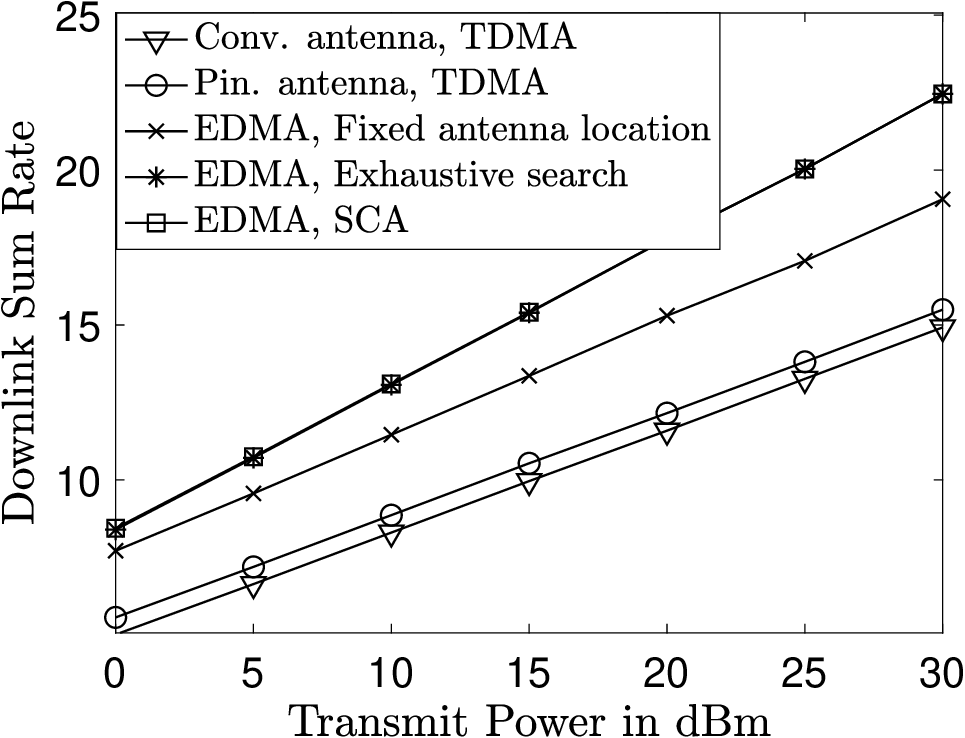}} 
\subfigure[$M=3$ and $D_{\rm L}=D_{\rm W}$ ]{\label{fig8b}\includegraphics[width=0.38\textwidth]{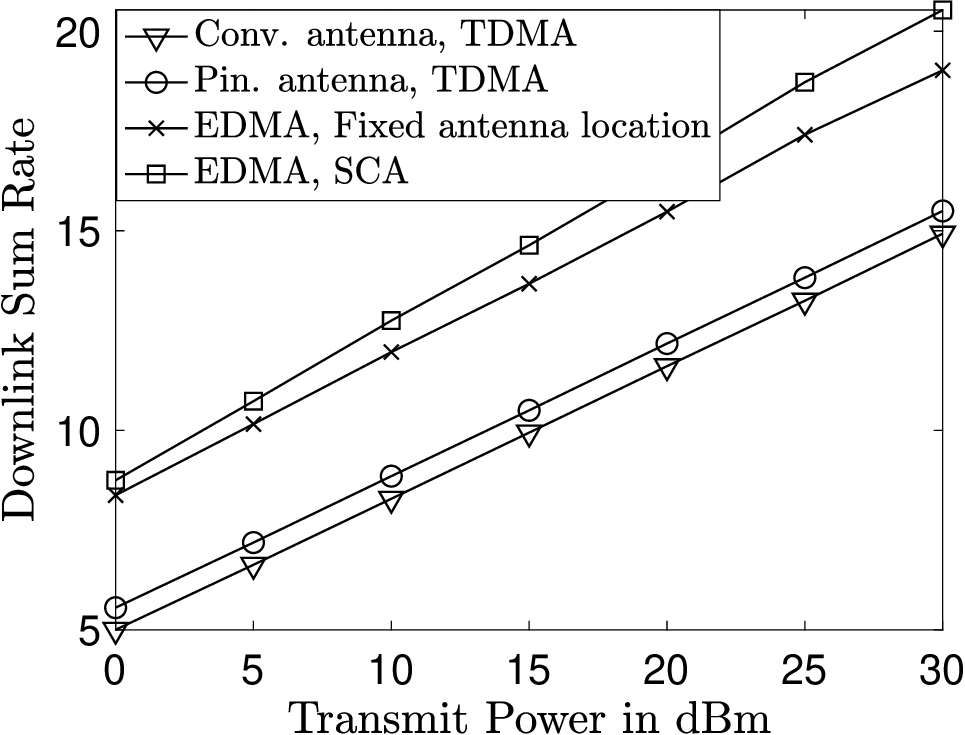}}   \vspace{-1em}
\end{center}
\caption{Impact of antenna location optimization on the downlink sum rate achieved by EDMA, where $\phi=0.02$ and the fixed antenna location scheme is based on $x_m^{\rm Pin}=x_m$.  The users are equally spaced in a statistical sense, similar to the scenario considered in Fig. \ref{fig7a}. \vspace{-1em} }\label{fig8}\vspace{-0.2em}
\end{figure}

In Fig. \ref{fig8}, the impact of antenna location optimization on the downlink sum rate achieved by EDMA is investigated, where the users are equally spaced in a statistical sense, a scenario similar to the one considered in Fig. \ref{fig7a}. In particular, Fig. \ref{fig8a} focuses on the case with $M=2$, where an exhaustive search is feasible due to the small search space. As can be seen from Fig. \ref{fig8a}, the performance of the proposed SCA algorithm is identical to that of the exhaustive search, which demonstrates the near-optimal performance of the proposed SCA scheme for this considered case. The use of the fixed location choice, $x_m^{\rm Pin}=x_m$, results in a performance loss compared to the SCA algorithm, as shown in Fig. \ref{fig8a}. When the number of users increases and the size of the service area stays the same, the use of SCA can still outperform the fixed-location scheme, but the performance gap between the SCA and the fixed-location scheme is reduced, as shown in Fig. \ref{fig8b}. Fig. \ref{fig9} shows that by increasing the size of the service area, the curves for the schemes with fixed antenna locations and the exhaustive search tend to coincide, which means that the performance of the fixed-location scheme becomes almost optimal if the users become widely separated. In addition, Fig. \ref{fig9} shows that for different choices of the service area size and the SNR, the proposed SCA approach consistently attains near-optimal performance.

     \begin{figure}[!]\centering \vspace{-0.2em}
    \epsfig{file=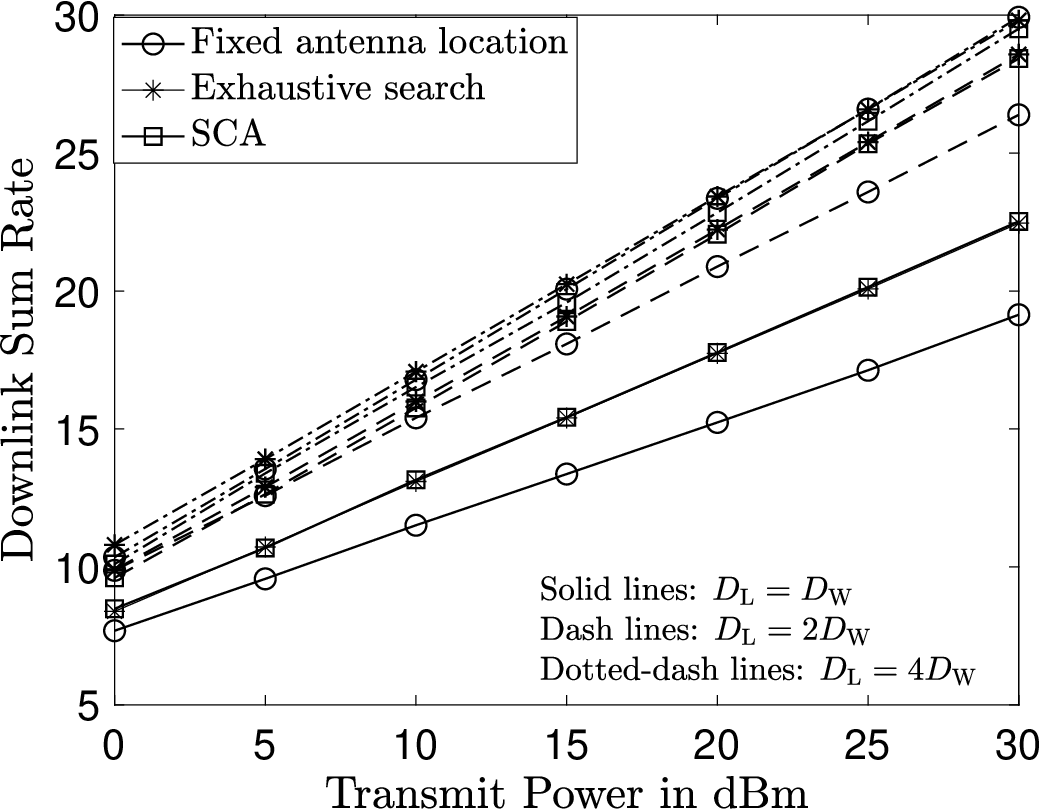, width=0.38\textwidth, clip=}\vspace{-0.5em}
\caption{The impact of the service area size on performance of EDMA, where $M=2$, and $\phi=0.02$. 
  \vspace{-1em}    }\label{fig9}   \vspace{-0.5em} 
\end{figure}

   \begin{figure}[!] \vspace{-0.2em}
\begin{center}
\subfigure[$D_{\rm L}=D_{\rm W}$ ]{\label{fig10a}\includegraphics[width=0.38\textwidth]{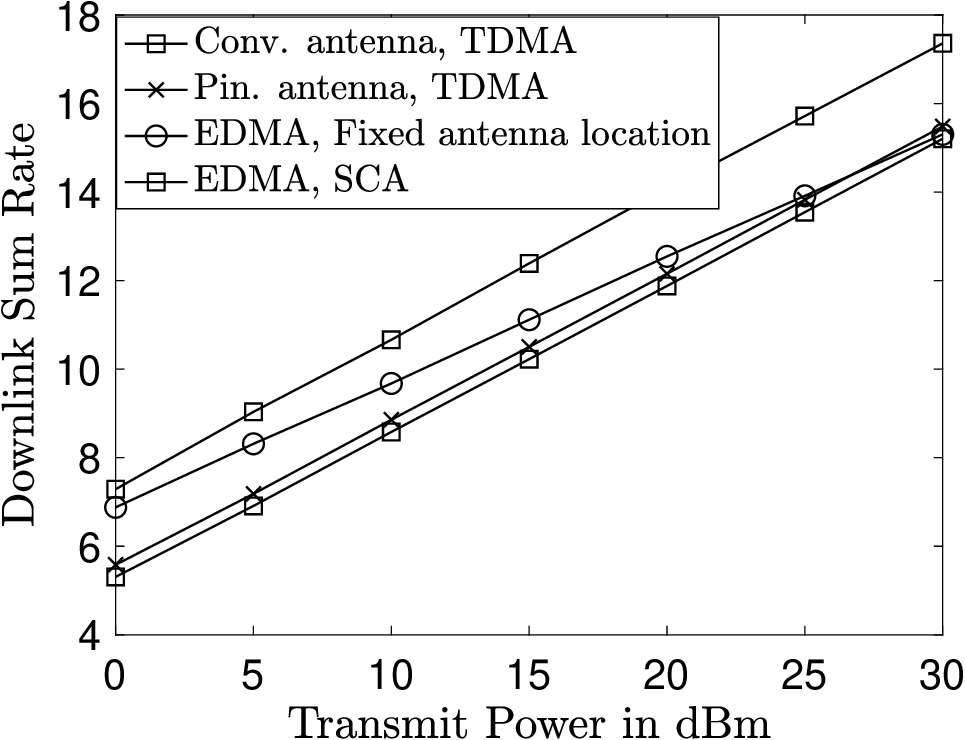}} 
\subfigure[ $D_{\rm L}=2D_{\rm W}$ ]{\label{fig10b}\includegraphics[width=0.38\textwidth]{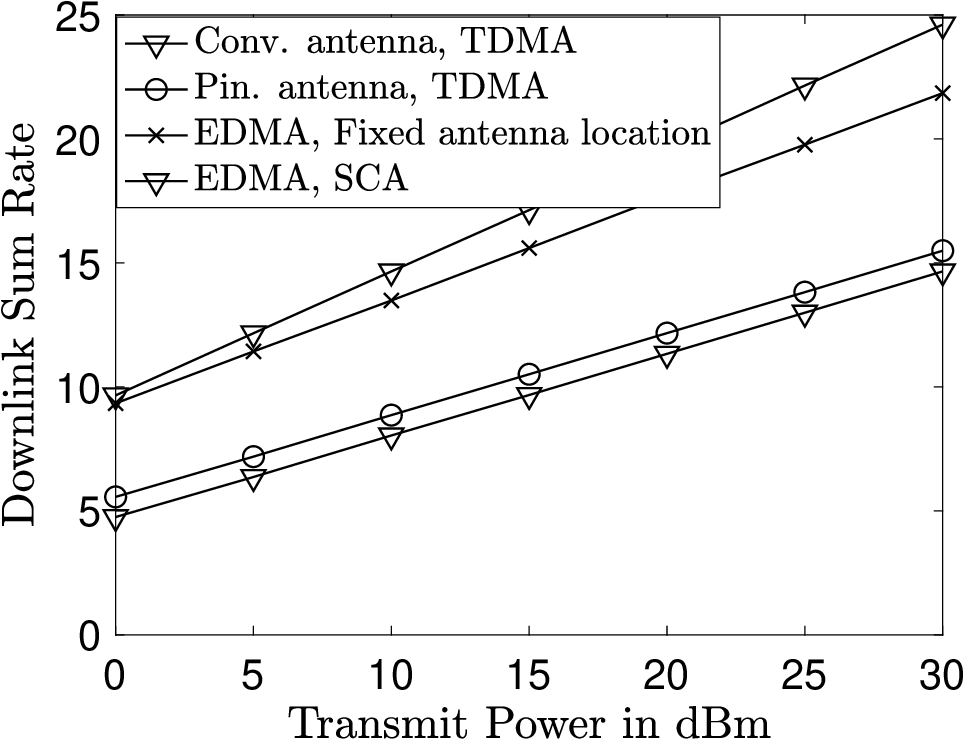}}   \vspace{-1em}
\end{center}
\caption{Impact of antenna location optimization on the downlink sum rate achieved by EDMA, where $M=3$ and $\phi=0.02$. The users are clustered as in the scenario considered in \ref{fig7b}.   \vspace{-1em} }\label{fig10}\vspace{-1.2em}
\end{figure}

Fig. \ref{fig10} considers the scenario in which the users are deployed close to each other, i.e., a scenario with clustered users similar to the one considered in Fig. \ref{fig7b}. Comparing Fig. \ref{fig8b} to Fig. \ref{fig10a}, one can observe that the performance of the fixed-location scheme is degraded severely and is even worse than that of the TDMA benchmark, since the users are close to each other and hence co-channel interference becomes severe. The use of SCA can effectively combat such co-channel interference, and there is still a significant performance gain over TDMA, as shown in Fig. \ref{fig10a}. When the size of the service area increases, both SCA and the fixed location scheme can ensure that the performance gap between EDMA over TDMA increases, as shown in Fig. \ref{fig10b}.  

\section{Conclusions}\label{section conclusions}
In this paper, a new multiple access technique, termed EDMA, has been proposed which exploits the dynamic features of wireless propagation environments. In particular, pinching-antenna-assisted EDMA has been focused on, where the multi-user propagation environment is intelligently reconfigured to improve the signal strength at intended receivers while simultaneously suppressing multiple-access interference. The key to creating a favorable propagation environment is to utilize the capability of pinching antennas for reconfiguring LoS links. Closed-form expressions for the ergodic sum-rate gain of EDMA over conventional multiple access have been developed to demonstrate the great potential of EDMA for supporting multi-user communications without requiring complex signal processing, e.g., precoding or multi-user detection. Furthermore, pinching antenna location optimization has been considered, where two low-complexity algorithms were developed for the uplink and downlink cases, respectively.  

In this paper, only LoS links were considered in the EDMA system model, given the fact that they are dominant over non-LoS (NLoS) links. Nevertheless, an important direction for future research is to investigate the performance and the system design of EDMA in the presence of NLoS links. In addition, the problem of optimizing the pinching antenna locations was solved by iterative methods, e.g., the golden section search algorithm and SCA. An important direction for future research is to obtain closed-form optimal solutions, which can provide a more insightful understanding of the performance of EDMA transmission \cite{closedformzid}. 

\appendices
\section{PROOF OF LEMMA \ref{lemma2}}\label{prooflemma2}

 It is straightforward to show that $g_1(a)$ is a monotonically increasing function of $a$, which means that $g\left( d^2+\rho\eta\right)-g\left( d^2\right)$ is always non-negative. Therefore, the lemma can be proved by focusing on the term $T_2=1-2e^{-\phi d^2  } T_1 \frac{\sqrt{\pi}}{D_{\rm W}\sqrt{\phi}} \Phi\left( \frac{\sqrt{\phi}D_{\rm W}}{2}\right)   $, which can be expressed as follows:
\begin{align}
T_2=&1-2e^{-\phi d^2  }   \frac{\sqrt{\pi}}{D_{\rm W}\sqrt{\phi}} \Phi\left( \frac{\sqrt{\phi}D_{\rm W}}{2}\right)  \\\nonumber &\times \frac{2\sqrt{\pi}}{D_{\rm L}\sqrt{\phi}} \left[\Phi(\sqrt{\phi} D_{\rm L})-\Phi \Big(\frac{\sqrt{\phi} D_{\rm L}}{2}\Big)\right]\\ \label{T1xx} &
+\frac{2}{\phi D_{\rm L}^2} \left(1-2e^{-\frac{\phi D_{\rm L}^2}{4}}+e^{-\phi D_{\rm L}^2}\right) .
\end{align} 
Recall that  for $x\rightarrow \infty$, the probability integral $\Phi(x)$ can be approximated as follows: \cite{GRADSHTEYN}
  \begin{align}
 \Phi(x) \approx 1-\frac{e^{-x^2}}{\sqrt{\pi}x}\left[
 \sum^{n}_{k=0}(-1)^k\frac{(2k-1)!!}{(2x^2)^k}
 \right] \approx 1-\frac{e^{-x^2}}{\sqrt{\pi}x} .
 \end{align}
 By using the above approximation, for the case with  $\sqrt{\phi} D_{\rm W}\rightarrow \infty$ and  $\sqrt{\phi} D_{\rm L}\rightarrow \infty$, $T_2$ can be approximated as follows:
 \begin{align}
T_2\approx &1-2e^{-\phi d^2  }   \frac{\sqrt{\pi}}{D_{\rm W}\sqrt{\phi}} 
\left(
1-\frac{e^{-\frac{ {\phi}D^2_{\rm W}}{4}}}{\sqrt{\pi}\frac{\sqrt{\phi}D_{\rm W}}{2}} 
\right)
  \\\nonumber &\times \frac{2\sqrt{\pi}}{D_{\rm L}\sqrt{\phi}} \left[
  \frac{e^{-\frac{ {\phi} D^2_{\rm L}}{4}}}{\sqrt{\pi}\frac{\sqrt{\phi} D_{\rm L}}{2}} -\frac{e^{- {\phi} D_{\rm L}^2}}{\sqrt{\pi}\sqrt{\phi} D_{\rm L}} 
 \right] 
+\frac{2}{\phi D_{\rm L}^2} \\\nonumber
\approx&1 
+\frac{2}{\phi D_{\rm L}^2}   \geq 0,
\end{align} 
where $e^{-x}\rightarrow 0$ for $x\rightarrow \infty$ was used. Because $T_2\geq 0$ and $g\left( d^2+\rho\eta\right)-g\left( d^2\right)\geq 0$, $\bar{\Delta}_{\rm sum}^{\rm LB}  \geq 0$ for the case of  $\sqrt{\phi} D_{\rm W}\rightarrow \infty$ and  $\sqrt{\phi} D_{\rm L}\rightarrow \infty$. 

On the other hand, in order to analyze the case with  $\sqrt{\phi} D_{\rm W}\rightarrow 0$ and  $\sqrt{\phi} D_{\rm L}\rightarrow 0$, we recall that for $x\rightarrow 0$, the probability integral $\Phi(x)$ can be approximated as follows: \cite{GRADSHTEYN}
 \begin{align}
 \Phi(x) =\frac{2}{\sqrt{\pi}}\sum^{\infty}_{k=1}(-1)^{k+1}\frac{x^{2k-1}}{(2k-1)(k-1)!}\approx \frac{2}{\sqrt{\pi}} x,
 \end{align}
 which means that $T_2$ can be approximated as follows:
 \begin{align}
T_2\approx &1-2e^{-\phi d^2  }   \frac{\sqrt{\pi}}{D_{\rm W}\sqrt{\phi}} 
\frac{2}{\sqrt{\pi}}  \frac{\sqrt{\phi}D_{\rm W}}{2}
  \\\nonumber &\times \frac{2\sqrt{\pi}}{D_{\rm L}\sqrt{\phi}} \left[
  \frac{2}{\sqrt{\pi}}   \sqrt{\phi} D_{\rm L}-
  \frac{2}{\sqrt{\pi}}   \frac{\sqrt{\phi} D_{\rm L}}{2}\right]\\ \label{T1dd} &
+\frac{2}{\phi D_{\rm L}^2} \left(1-2\left(1 -\frac{\phi D_{\rm L}^2}{4}\right)+1 -\phi D_{\rm L}^2\right) ,
\end{align} 
where the approximation $e^{-x}\approx 1-x$ for small $x$ is used.  With some straightforward algebraic manipulations, $T_2$ can be approximated as follows:
  \begin{align}
T_2\approx &-4e^{-\phi d^2  }  \leq 0,
\end{align} 
which means that $\bar{\Delta}_{\rm sum}^{\rm LB} \leq 0$ for the case with  $\sqrt{\phi} D_{\rm W}\rightarrow 0$ and  $\sqrt{\phi} D_{\rm L}\rightarrow 0$. Therefore, the proof for the lemma is complete. 

\section{PROOF OF LEMMA \ref{lemma4}}\label{prooflemma4}

In order to obtain the condition under which $f(x)$ is a unimodal function, we first decompose $f(x)$ as follows: 
\begin{align}
f(x) = & \gamma_1(x)   \log_2\left(
1+  \frac{\rho\eta}{ (x-x_{m})^2+ y_m^2 +d^2}
\right),
\end{align}
 where
  \begin{align}
 \gamma_1(x)=& \left(1- e^{-\phi\left((x_{m+1}-x )^2 +d^2\right) }  \right) \left(1-e^{-\phi\left( (x-x_{m-1})^2+d^2 \right)}\right). 
 \end{align} 
The proof can be completed by first showing that the term $\log_2\left(
1+  \frac{\rho\eta}{ (x_m^{\rm Pin}-x_{m})^2+ y_m^2 +d^2}
\right)$ is insignificant to the unimodality analysis at high SNR, and then studying the unimodal property of $
\gamma_1(x)$.
\subsubsection{The impact of the high-SNR assumption}  The unimodality of $f(x)$ can be equivalently studied by focusing on $\log f(x)$, whose first-order derivative can be expressed as in \eqref{dlogfx} at the top of the next page. 
\begin{figure*}\vspace{-2em}
\begin{align}\nonumber
\frac{d \log {f}(x)}{dx} 
=&\frac{ \gamma_1'(x) }{ \gamma_1(x) }-\frac{1}{\ln 2}   \frac{2\rho \eta(x-x_m)}{\big[\big((x-x_m)^2+d^2\big)^2 + \rho \eta ((x-x_m)^2+d^2)\big] \log_2\left(
1+  \frac{\rho\eta}{ (x-x_{m})^2  +d^2}
\right)} 
\\\label{dlogfx}
\approx&  \frac{ \gamma_1'(x) }{ \gamma_1(x) } -\frac{1}{\log_2 
 (\rho \eta)\ln 2}   \frac{2 (x-x_m)}{ ((x-x_m)^2+d^2)  
 }.
\end{align}\vspace{-2em}
\end{figure*}

On the one hand,  \eqref{dlogfx} shows that at high SNR, i.e., $\rho\rightarrow \infty$, $\frac{1}{\log_2 
 (\rho \eta)\ln 2}   \frac{2 (x-x_m)}{ ((x-x_m)^2+d^2)  
 }\rightarrow 0$. On the other hand, we note that $ \gamma_1(x)$ is not a function of $\rho$, which means that the term $\frac{ \gamma_1'(x) }{ \gamma_1(x) } $ is the dominant one in \eqref{dlogfx} at high SNR, i.e., $ \gamma_1(x)$ is the dominant term to decide the unimodality of $f(x)$.
 
 
 \subsubsection{Analyzing the unimodality of $\gamma_1(x)$}
 To facilitate the performance analysis, define ${x}_0=\frac{x_{m-1}+x_{m+1}}{2}$, $\delta=\frac{x_{m+1}-x_{m-1}}{2}$, and $x-x_0=w$, which means that $ \gamma_1(x)$ can be rewritten as follows: 
 \begin{align}\nonumber
 \gamma_1(w)=& \left(1- e^{-\phi\left((\delta-w )^2 +d^2\right) }  \right)   \left(1-e^{-\phi\left( (\delta+w)^2 +d^2 \right)}\right)\\   &= g(\delta-w)g(\delta+w),
  \end{align}
for $-\frac{\delta}{2}\leq w\leq \frac{\delta}{2}$, where $g(y)$ is defined as follows:
\begin{align}
g(y)=& \left(1- e^{-\phi\left(y^2 +d^2\right) }  \right)  .
\end{align}
We note that $ \gamma_1(w)$ is an even function. Therefore, in order to show that $\gamma_1(x)$ is unimodal with a single maximum for $x_{m-1}\leq x\leq x_{m+1}$, it is equvalent to show that $ \gamma_1(w)$ is a monotonically increasing function for $-\frac{\delta}{2}\leq w< 0$, and a monotonically decreasing function for $0\leq w\leq \frac{\delta}{2}$. Due to space limitations, the monotonicity of $\gamma_1(w)$ for $0\leq w\leq \frac{\delta}{2}$ is focused on in the following.
 
It is straightforward to verify that the monotonicity property of $\gamma_1(w)$ is the same as that of $\log \gamma_1(w)$, which can be expressed as follows: 
\begin{align}
\log \gamma_1(w)=&\log \left(1- e^{-\phi\left((\delta-w )^2 +d^2\right) }  \right) \\\nonumber &+ \log \left(1-e^{-\phi\left( (\delta+w)^2 +d^2 \right)}\right).
\end{align}
The first-order derivative of $\log \gamma_1(w)$ is given by 
\begin{align}
\frac{d \log \gamma_1(w)}{dw}=&- \frac{2\phi (\delta-w) e^{-\phi\left((\delta-w )^2 +d^2\right) } }{1- e^{-\phi\left((\delta-w )^2 +d^2\right) } } \\\nonumber &+\frac{2\phi (\delta+w)e^{-\phi\left( (\delta+w)^2 +d^2 \right)}}{1-e^{-\phi\left( (\delta+w)^2 +d^2 \right)}}\\\nonumber
=&- \frac{2\phi (\delta-w)   }{ e^{\phi\left((\delta-w )^2 +d^2\right) } -1}  +\frac{2\phi (\delta+w) }{ e^{\phi\left( (\delta+w)^2 +d^2 \right)}-1}.
\end{align}
To show that $\gamma_1(w)$ is a monotonically decreasing function for $0\leq w\leq \frac{\delta}{2}$, it is sufficient to show that $\frac{d \log \gamma_1(w)}{dw}$ is always non-positive, which means that the following inequality needs to be established:
\begin{align}\label{middl1}
 \frac{2\phi (\delta-w)   }{ e^{\phi\left((\delta-w )^2 +d^2\right) } -1}  \geq\frac{2\phi (\delta+w) }{ e^{\phi\left( (\delta+w)^2 +d^2 \right)}-1}. 
\end{align}
The above inequality can also be equivalently expressed as follows:
\begin{align}\label{middl2}
\gamma_3(\delta-w)\geq \gamma_3(\delta+w),
\end{align}
where $\gamma_3(y)= \frac{ y  }{ e^{\phi\left(y^2 +d^2\right) } -1}  $. 

Recall that  $0\leq w\leq \frac{\delta}{2}$, which means that $\delta-w\leq \delta+w$. Therefore, a sufficient condition for the inequality in \eqref{middl2} is that $\gamma_3(y)$
is a decreasing function of $y$, for $ \frac{\delta}{2}\leq y\leq  \frac{3}{2}\delta $. In order to analyze the monotonicity property of $\gamma_3(y)$,  $\gamma_3(y)$ is frist rewritten as follows:
\begin{align}\label{log3y}
  \gamma_3(y)=   \frac{ y  }{ h(s)}  ,
\end{align}
where $h(s)=e^{\phi s   } -1$ and $s=y^2 +d^2$. 
The first-order derivative of $  \gamma_3(y)$ is given by
\begin{align}
(  \gamma_3(y))'=& \frac{  1  }{ h(s)}  - \frac{  y  }{ h^2(s)} h'(s)s'\\\nonumber
=& \frac{  1  }{ h(s)}  -\frac{  2y^2  }{ h^2(s)} h'(s) .
\end{align}
For $\gamma_3(y)$ to be a decreasing function, $(  \gamma_3(y))'$ should be non-positive, i.e.,  
\begin{align}
  \frac{  1  }{ h(s)}  \leq&  \frac{  2y^2  }{ h^2(s)} h'(s)  ,
\end{align}
which can be shown by the following inequality related to $s$ only:
\begin{align} 
  1 \leq&  \frac{  2 (s-d^2)  }{ h(s)} h'(s)   .
\end{align}
By using the first-order derivative of $h(s)$, the above inequality can be expressed as follows:
\begin{align}  
  \frac{  2 (s-d^2)  }{ 1-e^{-\phi s   } }  \phi    &\geq 1.
\end{align}
Recall that $s\geq \frac{\delta^2}{4}+d^2$. Therefore, the term $  \frac{  2 (s-d^2)  }{ 1-e^{-\phi s   } }  \phi  $ can be lower bounded as follows: 
\begin{align}  
   \frac{  2 (s-d^2)  }{ 1-e^{-\phi s   } }  \phi  \geq   2 (s-d^2)   \phi  \geq    \frac{\delta^2\phi}{2}    .
\end{align}
 
Therefore, if $\frac{\delta^2\phi}{2}  \geq 1$, $\gamma_3(y)$ is a decreasing function, and hence $\frac{d   \gamma_1(w)}{dw}$ is always negative, which means that $\gamma_1(w)$ is indeed a monotonically decreasing function for $0\leq w\leq \frac{\delta}{2}$. The fact that $\gamma_1(w)$ is a monotonically increasing function for $-\frac{\delta}{2}\leq w\leq 0$ can be proved similarly. The proof of the lemma is complete. 

  \vspace{-0.5em}
\bibliographystyle{IEEEtran}
\bibliography{IEEEfull,trasfer}
  \end{document}